\documentclass[11pt]{article}
\usepackage[T1]{fontenc}
\usepackage[utf8]{inputenc}
\usepackage{mathpazo}
\usepackage[a4paper,margin=1in]{geometry}
\usepackage{amsmath,amssymb,amsthm}
\usepackage{bbm}
\usepackage{thm-restate}
\usepackage{graphicx}
\usepackage{float}
\usepackage{booktabs}
\usepackage{enumitem}
\usepackage{array}
\usepackage{tabularx}
\usepackage{ragged2e}
\usepackage{xcolor}
\usepackage{tikz}
\usepackage{changepage}
\usepackage[numbers]{natbib}
\usepackage[colorlinks=true,linkcolor=blue,citecolor=blue,urlcolor=blue]{hyperref}
\usepackage[nameinlink,noabbrev]{cleveref}

\definecolor{darkRed}{HTML}{EF3B2C}
\definecolor{darkBlue}{HTML}{2166AC}

\newtheorem{theorem}{Theorem}[section]
\newtheorem{lemma}[theorem]{Lemma}
\theoremstyle{definition}
\newtheorem{definition}[theorem]{Definition}
\newenvironment{Theorem}{\begin{theorem}}{\end{theorem}}
\newenvironment{Lemma}{\begin{lemma}}{\end{lemma}}
\newenvironment{Definition}{\begin{definition}}{\end{definition}}
\crefname{theorem}{Theorem}{Theorems}
\Crefname{theorem}{Theorem}{Theorems}
\crefname{lemma}{Lemma}{Lemmas}
\Crefname{lemma}{Lemma}{Lemmas}
\crefname{definition}{Definition}{Definitions}
\Crefname{definition}{Definition}{Definitions}

\title{Stable Boundaries of Opinion Dynamics in Heterogeneous Spatial Complex Networks}
\author{%
Mats Bierwirth\\
Department of Computer Science, ETH Z\"urich, Switzerland\\
\texttt{mbierwirth@student.ethz.ch}
\and
Johannes Lengler\\
Department of Computer Science, ETH Z\"urich, Switzerland\\
\texttt{johannes.lengler@inf.ethz.ch}}
\date{}

\begin{document}
\maketitle

\begin{abstract}
We investigate majority-vote opinion dynamics on Geometric Inhomogeneous Random Graphs (GIRGs), a powerful model for spatial complex networks. In contrast to classic coarsening dynamics where a single opinion typically achieves global consensus, our simulations reveal that sufficiently large, localized opinion domains do not disappear. Instead, they stabilize, leading to a persistent coexistence of competing opinions. To understand the mechanism behind this arrested coarsening, we develop and analyze a tractable mean-field model of the interface between two opinion domains. Our main theoretical result rigorously establishes the existence of a stable, non-trivial limiting distribution for the interface profile in a mean-field analysis. This demonstrates that the boundary between opinions is stationary, providing a mathematical explanation for how complex network geometry can support robust opinion diversity in social systems.
\end{abstract}


\section{Introduction}
\noindent\textbf{Emergence of Stable Boundaries in Spatial Opinion Dynamics.} The diffusion of ideas, the adoption of innovations, and the formation of political consensus are all driven by social influence within populations structured by complex networks. A growing body of evidence suggests that many real-world networks, from online social platforms to offline communities, can be modelled by assuming a latent geometric space where proximity increases the connection probability~\cite{boguna2010sustaining,bringmann2019geometric,blasius2024external,blasius2024robust,dayan2024expressivity,komjathy2024polynomial}. This underlying geometry gives rise to many properties observed in real social networks, including high clustering and stronger communities than the degree distribution would predict~\cite{bringmann2019geometric,blasius2024robust,kaufmann2024sublinear,kaufmann2025assortativity}. Naturally, those structural properties shape the dynamics of social processes in the network.

A fundamental social mechanism is conformism: individuals tend to adopt the opinion held by the majority of their peers~\cite{Flache2017}. This principle is captured by the simple and intuitive majority-vote dynamics, a process where agents iteratively update their state to match their local environment~\cite{Galam2002,Krapivsky2003,Castellano2009} and maintain their state in case of a tie. This work begins with a simple yet foundational question: in a world divided into two competing opinions, what determines the fate of a localized enclave of one opinion surrounded by the other? To investigate this, we performed simulations of a sequential majority process on Geometric Inhomogeneous Random Graphs (GIRGs), a state-of-the-art model for spatial complex networks~\cite{Bringmann2017,bringmann2019geometric}. The simulations then lead to a hypothesis that we rigorously verify in a mean-field approximation of the process and the graph models.

Our simulations reveal a dichotomy that forms the central puzzle of this paper. When an initial, localized domain of one opinion--say, ``blue''--is small, it is quickly eroded by the surrounding ``red'' majority and vanishes, leading to global consensus. This outcome aligns with the classical expectation of coarsening dynamics~\cite{bray2002theory}. However, if the initial blue domain is sufficiently large, the dynamics are markedly different. The domain initially shrinks and its boundaries become smoother, but the process halts. The system settles into a stable configuration where a persistent, ball-like cluster of the blue opinion coexists indefinitely with the red majority. This phenomenon of arrested coarsening suggests that interfaces between opinion domains in these complex networks can be remarkably stable. The central question of this paper is, therefore: what is the underlying mechanism that governs this stability?\smallskip

\noindent\textbf{Context: From Coarsening Physics to Complex Networks.} The evolution of boundaries between competing phases is a classic topic in statistical physics~\cite{bray2002theory}. In models such as the Voter Model on regular Euclidean lattices, where vertices randomly adopt the state of one of their neighbours, the system typically evolves to minimize the length of the interface separating domains, a process known as coarsening~\cite{Dornic2001}. On any finite graph, this process will continue until eventually one domain has been completely eliminated, resulting in a global consensus where all agents share the same state. In many graphs, such convergence happens very quickly~\cite{cooper2010multiple,lyons2017probability,cox1989coalescing}. For majority dynamics as studied in this paper, consensus is not guaranteed, but has also been shown to emerge quickly on many graphs~\cite{gartner2018majority,mossel2014majority}. Our observation represents a significant departure from this classical picture and points toward a mechanism that is intrinsic to the underlying network. While stable coexistence of two opinions for majority dynamics trivially occurs in simple spatial graphs like lattices, it was unclear whether complex networks with long-range edges would show the same pattern. 

The network in our study is the Geometric Inhomogeneous Random Graph (GIRG) model~\cite{bringmann2019geometric}. GIRGs have emerged as a powerful and realistic framework for modelling complex networks~\cite{jorritsma2020not,komjathy2024polynomial,komjathy2023four,blasius2024robust,blasius2024external,kaufmann2025assortativity}. They rely on two ingredients. First, they are embedded in a geometric space, with connection probability decaying with distance. Second, vertices are endowed with heterogeneous weights, drawn from a power-law distribution, which gives rise to a scale-free degree distribution and the existence of highly connected ``hub'' vertices. The combination of those two features gives rise to many other emerging properties of real-world networks: they have clustering and communities~\cite{bringmann2019geometric}, ultra-small distances~\cite{bringmann2025average}, they are navigable~\cite{bringmann2017greedy} and  compressible~\cite{bringmann2019geometric}. They show a remarkably rich phase diagram for infection processes~\cite{komjathy2024polynomial,komjathy2023four} and rumour spreading~\cite{kaufmann2026rumour}, and the performance of several graph algorithms on GIRGs has been shown to match closely that on real-world networks~\cite{blasius2024external,cerf2024balanced}.
They also include the popular model of Hyperbolic Random Graphs~\cite{Krioukov2010} as special case~\cite{bringmann2025average}. While many spreading processes like infection models~\cite{jorritsma2020not,komjathy2024polynomial}, rumour spreading~\cite{kaufmann2026rumour}, first-passage percolation~\cite{komjathy2020explosion,komjathy2021penalising}, and bootstrap percolation~\cite{Koch2016} have been analysed for GIRGs, to the best of our knowledge this is the first paper studying a competitive dynamics between two opinions. \smallskip

\noindent\textbf{Mathematical Contribution: A Tractable Mean-Field Model of the Interface.} It is not hard to understand why a too small ball of diverging opinion disappears. We give a brief heuristic argument in Section~\ref{sec:experiments} that a boundary is unstable when its curvature becomes too small. This also explains the experimentally observed phenomenon that a large box of blue opinion will slightly shrink into a ball before stabilizing. However, the mathematically more challenging part is to understand why a sufficiently large ball stabilizes.

To gain mathematical traction on this question, we analyse the process in a simplified, continuous setting. Since stability of the interface is only achieved for sufficiently large balls, we first consider the macroscopic limit of an infinitely large ball of one opinion. In this limit, the curvature of the interface approaches zero, and the boundary converges to a straight line separating two half-spaces, one initially all-red and the other all-blue. This approach of studying a planar interface is a powerful technique in the statistical mechanics of phase separation and allows a detailed analysis of the interface profile and its dynamics. In a second step we then show how the results for this macroscopic limit transfers to the case of balls with large but finite radius.

We model the state of this idealized system by a function $f(z)$, which represents the probability that a vertex located at a signed distance $z$ from the interface holds the blue opinion. The evolution of $f(z)$ over time is governed by an operator derived from the majority dynamics. In each step of our mean-field process, a vertex' neighbourhood is effectively resampled from a hypothetical population whose red-/blue-assignment follows the probability distribution of the previous step. The number of red and blue neighbours of a vertex are then independent Poisson random variables, redrawn in each round. This is a common and effective approximation for sparse random graphs where the presence of edges is largely independent, but it is a substantial simplification for spatial models. The expectation of the two Poisson random variables is then obtained by integrating the current opinion profile $f$ against the connection kernel of the GIRG model. This ``mean-field assumption'', formally stated in \Cref{def:mean-field} of our technical analysis, decouples local spatial correlations. While this ignores the fact that neighbours of a red vertex are themselves more likely to be red, it makes the system analytically tractable. The evolution of the system is then described by an update operator $\mathcal T$ derived from the dynamics.\smallskip

\noindent\textbf{Our Results.} Our main theoretical result is concerned with the limiting distribution $f^* = \lim_{i\to \infty} \mathcal T^if$ obtained by applying the operator $\mathcal T$ repeatedly. In principle, a mean-field dynamics as described above could lead to two types of limiting distributions: the constant function $\hat f \equiv 1/2$ is a fixed point of $\mathcal T$, so that would be a candidate for $f^\star$, which would then indicate that the phase boundary is not stable. However, we show that this is not the case in the macroscopic limit and that the limiting distribution is of a different type. Our main theoretical result, presented in \Cref{thm:main}, demonstrates that the limiting distribution $f^\star(z)$ between halfspaces is bounded away from $1/2$ for sufficiently large $z$. In other words, the initially red region maintains a red majority forever, and likewise for the blue region. The existence of this stable, non-constant solution provides a rigorous mathematical explanation (in a simplified setting) for the persistence of the interface observed in our simulations. 

Our proof proceeds by constructing a class of ``valid'' functions that satisfy certain structural properties (monotonicity and symmetry) and that are subsolutions of $\mathcal T$, i.e., the operator $\mathcal T$ maps them pointwise into more extreme functions, meaning that the absolute distance from $1/2$ increases pointwise. We then show that any valid subsolution of $\mathcal T$ is a pointwise lower bound for the distance from $1/2$ of the true solution $f^\star$. Finally, for sufficiently large average degree in the underlying GIRG we show the existence of a non-trivial valid subsolution of $\mathcal T$ by an explicit construction, thereby guaranteeing that both opinions survive in their respective majority regions.

In a second step, we transfer the results to the case of balls of finite, but growing radius $r$, which is more subtle. Due to the finite curvature we can not expect the ball to persist indefinitely, as we formally show in Theorem~\ref{thm:non_convergence}. However, we show that the speed by which the ball erodes shrinks with the size of the ball. Note that the speed here is measured additively, i.e., if we start with a ball of radius $r = \omega(1)$ then even after time $t=\omega(1)$ the local opinion will still dominate in a ball of radius $r-o(1)$. This means that the speed of erosion approaches zero as $r\to \infty$. In the discussion we argue that this also explain arrested coarsening in the discrete setting of actual graphs. A vertex in a GIRG has typical geometric distance $\Omega(1)$ from all other vertices. Hence, in the discrete graph setting a ball of local opinion cannot shrink with arbitrarily small speed, and speed $o(1)$ in the mean field approximation naturally corresponds to speed zero in the discrete setting, meaning stability. \smallskip

\noindent\textbf{Summary.} In summary, this paper provides the first analytical evidence for the stability of opinion domains in the majority dynamics model on Geometric Inhomogeneous Random Graphs. By combining direct simulation with a rigorous mean-field analysis, we show that opinion forming in complex networks may deviate from the classic picture of opinion coarsening, with a limiting robust coexistence instead of global consensus. Our results shed light on the mechanisms that can support opinion diversity and the formation of stable ideological clusters in spatially-embedded social systems. Furthermore, they provide a new analytical framework for studying interface phenomena on complex networks, bridging concepts from statistical physics with the modern theory of random graphs.

A summary of the results has been presented at the 14th International Conference on Complex Networks and their Applications in Binghamptom, NY, US~\cite{bierwirth2025stable}. The associated proceedings paper contains an outline of the results, but without proofs. Moreover, it only contains results for the macroscopic limit where both opinions form halfspaces, whereas now we also show how those results transfer to balls of finite sizes. To extend the proof to this domain, we have switched the norm of the underlying geometric space from the maximum norm $\ell_\infty$ to the Euclidean norm $\ell_2$ because this makes the space rotationally invariant.

\section{Network Model and Opinion Spreading}\label{sec:definitions}
\noindent We start with the definition of Geometric Inhomogeneous Random Graphs (\emph{GIRG}) and state a few basic properties that will be useful in our analysis. 
\begin{Definition}[GIRG~\cite{bringmann2019geometric}]
    Let $d \in \mathbb{N}$, $\tau>2$ and depending on the previous $k$ be fixed constants and let $\mathcal{X}$ be a $d$-dimensional cube of volume $n \in \mathbb{N}$ with torus topology centered at the origin. Distances $\lVert x_u - x_v\rVert$ are measured with respect to the $\ell_2$ norm on this torus. A Geometric Inhomogeneous Random Graph $\mathrm{G} = (\mathrm{V},\mathrm{E})$ on $n$ vertices is obtained by the following three-step procedure.
    \begin{enumerate}[label=(\alph*)]
        \item Each vertex $v \in \mathrm{V}$ independently draws a weight $w_v$ from the power-law distribution on $\mathcal{D} = [1,\infty)$ with density
        $\rho(w) = (\tau-1) w^{-\tau}$ for $w \geq 1$.
        \item Each vertex $v \in \mathrm{V}$ draws independently a uniform random position $x_v\in \mathcal{X}$.
        \item A vertex pair $u, v \in \mathrm{V}$ forms an edge if and only if $\lVert x_u-x_v\rVert^d \le k \,w_u\,w_v$. 
    \end{enumerate}
\end{Definition}
\noindent Compared to~\cite{bringmann2019geometric}, we restrict to the zero-temperature case\footnote{In general, a temperature parameter $T = 1/\alpha$ makes edges appear with probability $\min\{1,w_u\, w_v / \lVert x_u - x_v \rVert_\infty^d\}^{\alpha}$. 
At zero temperature ($\alpha = \infty$), this becomes a deterministic threshold rule, which is analytically more tractable while still preserving the essential geometric and heterogeneous structure of GIRGs.}, 
and we explicitly model the \emph{density parameter} $k$ that was hidden in $\Theta$-notation. 
Equivalently to our use of $k$, we could also increase the density of vertices from $1$ to~$k$.

The \emph{ball of influence} is the region in which a vertex connects to all other vertices independently of their weights. It is a region that the vertex $v$ dominates.

\begin{Definition}[Ball of Influence]
    Let $\mathrm{G} = (\mathrm{V},\mathrm{E})$ be a \emph{GIRG} on $n$ vertices. The \emph{ball of influence} $I(v)$ of a vertex $v$ is the $\ell_2$-ball around $x_v$ of radius $r_{I(v)}:=k^{1/d} \, w_v^{1/d}$. Any vertex $u$ in $I(v)$ is a neighbour of $v$ regardless of $w_u$.
\end{Definition}
\noindent Beyond the ball of influence, vertices may still connect if their weight is sufficiently large. We call neighbours from the ball of influence \emph{near}-neighbours or $N^\mathrm{near}(v)$ for short, and other neighbours \emph{far}- neighbours or $N^\mathrm{far}(v)$ for short.

The following theorem gives the expected degree of a vertex of given weight. It is one of the fundamental properties of the model and will later serve as a key ingredient in our mean-field approximation. The following refines a result from \citet{bringmann2019geometric}.

\begin{Theorem} \label{thm:NumberOfNeighbours} Let $\mathrm{G} = (\mathrm{V},\mathrm{E})$ be a GIRG on $n$ vertices and  $v \in \mathrm{V}$. Condition on the event that $v$ has weight $w_v$. Then
\begin{align*}
    \mathbb{E}[|N(v)|\mid \,w_v ] = \frac{\pi^{d/2}}{\Gamma\left(\frac{d}{2}+1\right)} k \, w_v \, \left( 1 + \frac{1}{\tau-2}\right),
\end{align*}
where the first term is contributed by the near-neighbours and the second from the far-neighbours.
\end{Theorem}
\noindent We split the proof into two lemmas, corresponding to the two contributions to the degree: near neighbours and far neighbours. 
\begin{Lemma} \label{lem:neigh_ball}
    Let $\mathrm{G} = (\mathrm{V},\mathrm{E})$ be a \emph{GIRG} on $n$ vertices and  $v \in \mathrm{V}$. Condition on the event that $v$ has weight $w_v$, the expected number of neighbours of $v$ from the ball of influence is
    \begin{align*}
        \mathbb{E}\bigl[|N^\mathrm{near}(v)|\mid w_v\bigr] = \frac{\pi^{d/2}}{\Gamma\left(\frac{d}{2}+1\right)}  k \, w_v.
    \end{align*}
\end{Lemma}
\begin{proof}
    With respect to $\lVert .\rVert_2$, the ball of influence $I(v)$ is a $d$-dimensional ball $B{r_I}$ of radius $r_I$.
 The volume of such a hypercube is
    \begin{align*}
        \mathrm{Vol}(B_{r_I}) = \frac{\pi^{d/2}}{\Gamma\left(\frac{d}{2}+1\right)} \, r_I^d,
    \end{align*}
    where $\Gamma$ represents Eulers gamma function. 
    \noindent By definition of $r_I$ all vertices inside $B_{r_I}$, regardless of their weight, connect to $v$. Since $\mathrm{G}$ distributes $n$ vertices in a space of volume $n$, the density of the expected number of vertices per volume is one. Thus the ball of influence will contain one vertex for each volume of space contained in $B_{r_I}$. Thus
    \begin{align*}
        \mathbb{E}\bigl[N^\text{near}(v)\,\bigm|\,w_v\bigr] &= \text{Vol}(B_{r_I}) = \frac{\pi^{d/2}}{\Gamma\left(\frac{d}{2}+1\right)} k \, w_v \, \frac{1}{\tau-2},
    \end{align*}
    concluding the proof.
\end{proof}

\begin{Lemma} \label{lem:neigh_far}
    Let $\mathrm{G} = (\mathrm{V},\mathrm{E})$ be a GIRG on $n$ vertices and  $v \in \mathrm{V}$. Condition on the event that $v$ has weight $w_v$, the expected number of far neighbours of $v$ is
    \begin{align*}
        \mathbb{E}\bigl[N^\mathrm{far}(v)\bigm|\,w_v\bigr] &=  \frac{\pi^{d/2}}{\Gamma\left(\frac{d}{2}+1\right)} k \, w_v \, \frac{1}{\tau-2}.
    \end{align*}
\end{Lemma}
\begin{proof}
    Let $u$ be a far neighbour of $v$ with $r = \lVert x_u-x_v\rVert > r_I$. Recall that $u$ is a neighbour of $v$ if and only if $w_u  \geq r^d / (k \, w_v)$. 
    Thus the probability that a random vertex $u$ at distance $r$ from $v$ is a (far) neighbour of $v$ is
    \begin{align*}
        \Pr\left[ u\sim v \mid r \right] = \int_{\frac{r^d}{k \, w_v}}^{\infty} \Pr\left[w_u = w \right] \,dw.
    \end{align*}
    By definition, $\Pr\left[w_u = w \right]$ is drawn from a power-law distribution with known density. Thus
    \begin{align*}
    \Pr\left[ u\sim v\mid r  \right] &=  \int_{\frac{r^d}{k \, w_v}}^{\infty} (\tau-1)\,w^{-\tau} \,dw 
    = \left[-w^{1-\tau}\right]_\frac{r^d}{k \, w_v}^\infty 
    = \frac{r^{d \, {(1-\tau)}}}{k^{1-\tau} \, w_v^{1-\tau}}.
    \end{align*}
    For each fixed $r$, the space at distance exactly $r$ from $v$ forms the surface of a $d$-ball. Applying its volume definition and the same same volume argument from Lemma~\ref{lem:neigh_ball}, 
    \begin{align*}
        \mathbb{E}\bigl[|N^\text{far}s(v)|\bigm|\,w_v\bigr] &=\int_{r_I}^\infty  \frac{2\, \pi^{d/2}}{\Gamma\left(\frac{d}{2}\right)} r^{d-1}\,  \Pr\left[ u\sim v\mid r \right]\,dr \\
        &= \frac{2\, \pi^{d/2} \, k^{\tau-1}\, w_v^{\tau-1}}{\Gamma\left(\frac{d}{2}\right)\, d \, (\tau-2)} \left[-r^{d\,(2-\tau)}\right]_{r_I}^{\infty} \\
         &= \frac{\pi^{d/2} \, k^{\tau-1}\, w_v^{\tau-1}}{\Gamma\left(\frac{d}{2}+1\right) \, (\tau-2)} \left[-r^{d\,(2-\tau)}\right]_{r_I}^{\infty}
    \end{align*}
    the number of far neighbours can be expressed in terms of the probability $\Pr\left[ u\sim v\mid r  \right]$. The exponent $d(2-\tau)$ is negative, since $\tau>2$ and he integral evaluated at $\infty$ will therefore vanish. Plugging in $r_I = k^{1/d} \, w_v^{1/d}$ yields
    \begin{align*}
        \mathbb{E}\bigl[N^\text{far}(v)\bigm|\,w_v\bigr] &=  \frac{\pi^{d/2}}{\Gamma\left(\frac{d}{2}+1\right)} k \, w_v \, \frac{1}{\tau-2},
    \end{align*}
    concluding the proof.
    \end{proof}

\noindent On GIRGs we will consider the following sequential majority dynamics.
\begin{Definition}[Opinion Spreading]
Let $\mathrm{G} = (\mathrm{V},\mathrm{E})$ be an undirected graph. An \emph{opinion configuration} at time $t \in \mathbb{N}$ is a function $c_t : V \to \{-1, 1\}$, where $c_t(v)$ denotes the opinion of vertex $v \in V$ at time $t$. For given initial configuration $c_0$, the \emph{Opinion Spreading} process evolves as follows:  
at each time step $t \geq 1$, a vertex $v \in V$ is chosen uniformly at random. This vertex then converts to the majority opinion of its neighbourhood at time step $t-1$. The opinion of all other vertices remain unchanged. That is,
\begin{align*}
c_t(v) = 
\begin{cases}
    \mathrm{sign} \left( \sum\limits_{u \in N(v)} c_{t-1}(u) \right) & \text{if }  \sum\limits_{u \in N(v)} c_{t-1}(u) \neq 0, \\
    c_{t-1}(v) & \text{otherwise},
\end{cases}
\end{align*}
where $N(v)$ denotes the set of neighbours of $v$ in $G$. In particular, in the case of a tie, the opinion of $v$ remains unchanged.
\end{Definition}
\noindent It is known that the Opinion Spreading process always reaches a stable configuration $c^*$ in expected time at most $|V|\cdot|E|$, see~\cite{mossel2014majority}.

\section{Results}
\subsection{Experimental Observations}\label{sec:experiments}
\noindent To motivate the subsequent theoretical analysis, we conducted simulations of the Opinion Spreading process on GIRGs generated with the \texttt{libgirgs-all} library~\cite{GeneratingGirgs}. The experiments used $n=10\,000$ vertices in $d=2$ dimensions, that is a torus of side length $100$, and an average degree of $20$, while varying the degree exponent $\tau>2$. This parameter strongly influences the network structure. For small values of $\tau$, the graph tends to contain many high-weight vertices with long-range connections, whereas larger $\tau$ yields more localised networks with fewer hubs and fewer long edges. 

\begin{figure}[ht] 
    \centering
    \includegraphics[width=0.24\textwidth]{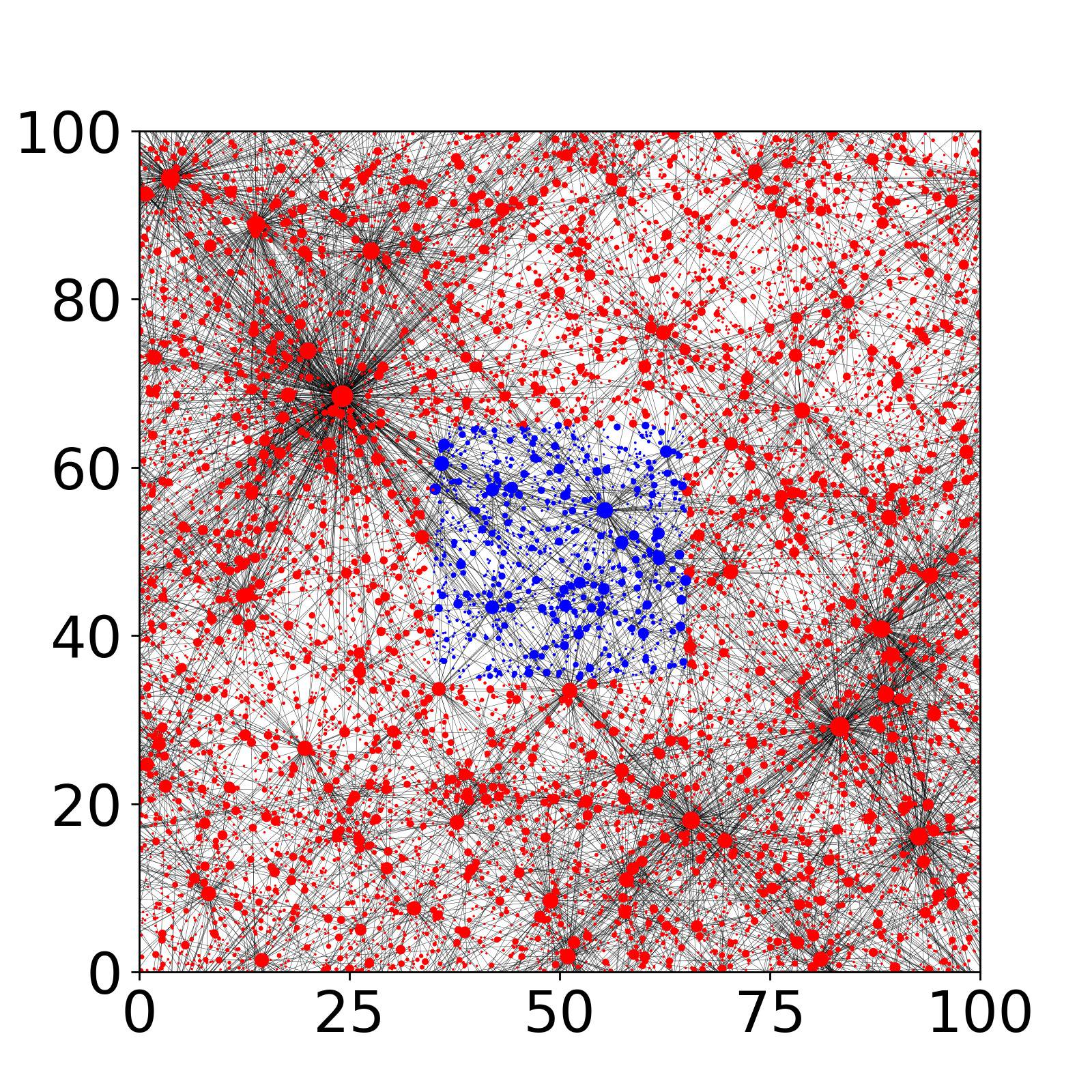} \includegraphics[width=0.24\textwidth]{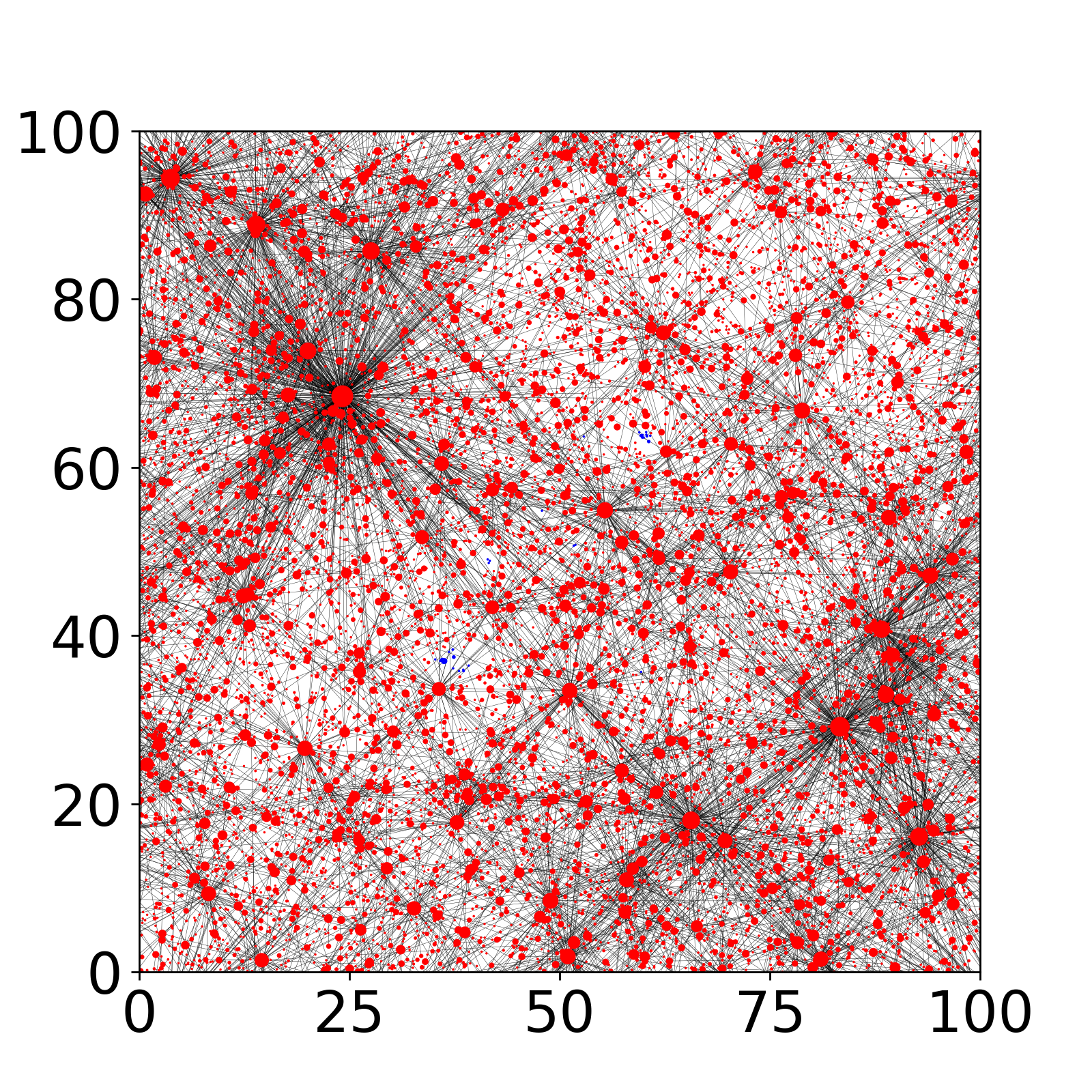}
    \includegraphics[width=0.24\textwidth]{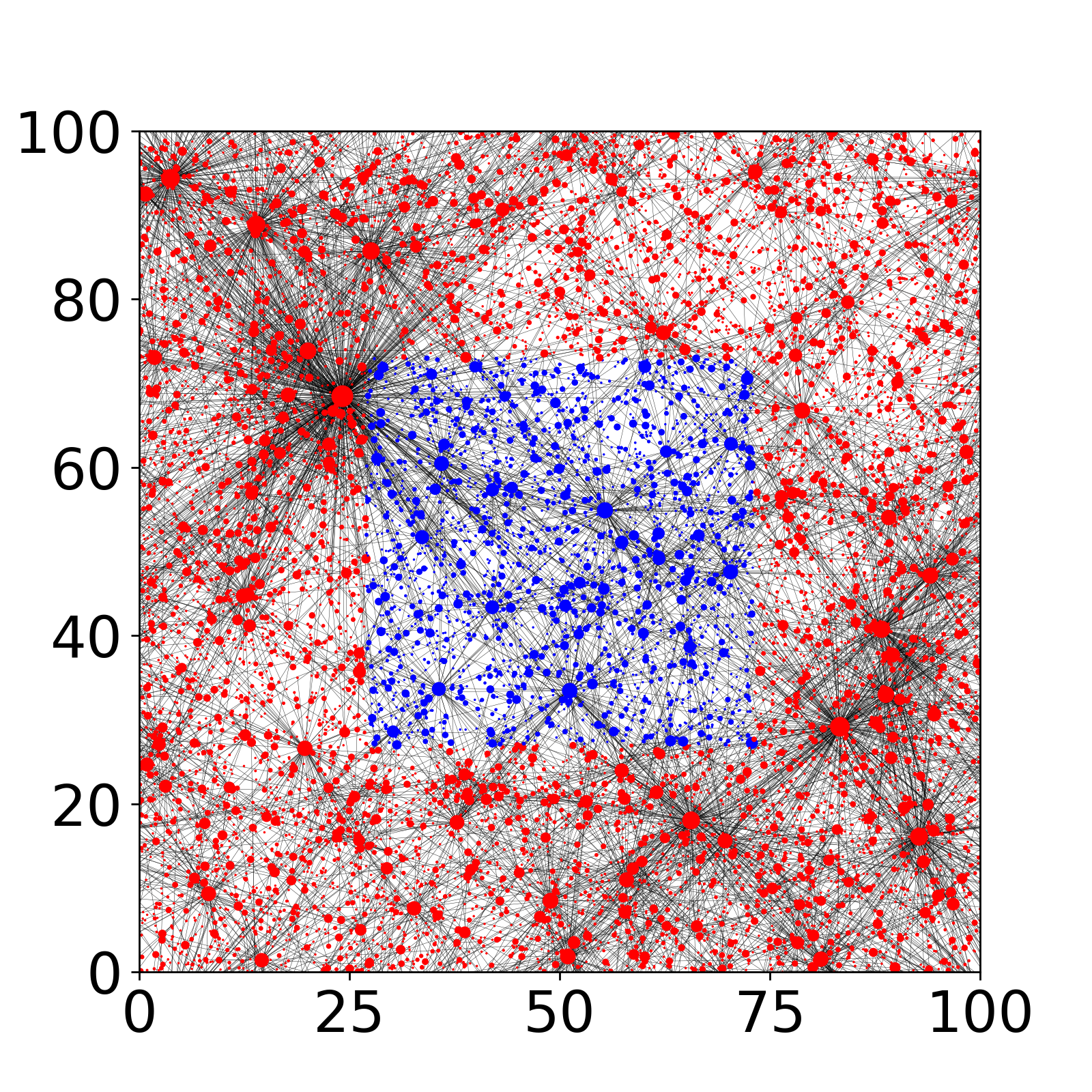} \includegraphics[width=0.24\textwidth]{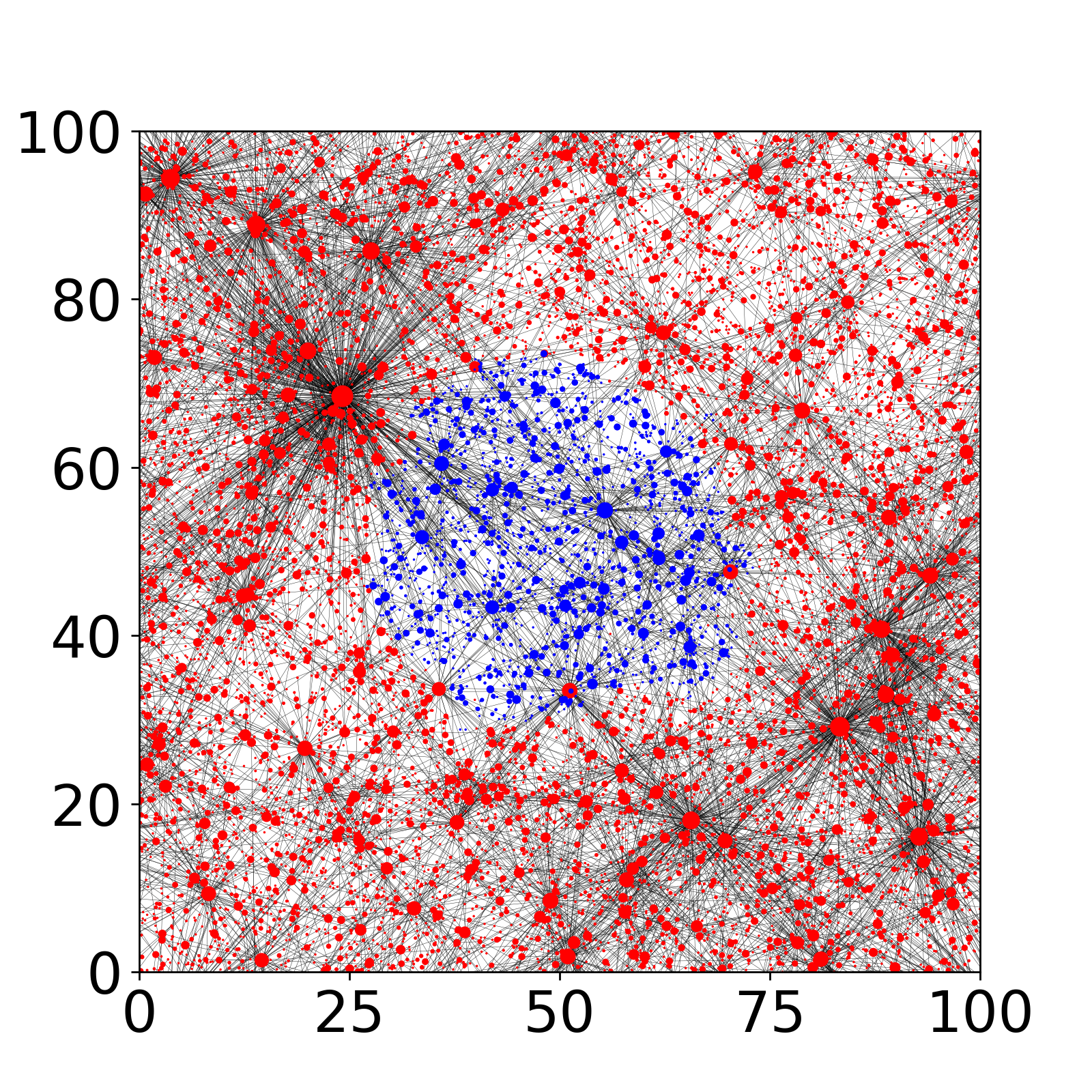}
    \caption{Opinion spreading with $\tau=2.15$. 
    Panels 1,2: small square initial configuration ($c_0$ left, final configuration $c^*$ right), only red survives with minor exceptions. 
    Panels 3,4: large square initial configuration ($c_0$ left, $c^*$ right), both opinions survive.}
    \label{fig:square-combined}
\end{figure}
\noindent The initial opinion configuration was chosen to contain a region of blue vertices ($+1$) in the shape of an axis-aligned square of side length $s$, centered at the origin, with all remaining vertices red ($-1$). 
When $s$ was small, the region of blue vertices did not survive. The red opinion expanded inward until the stable configuration $c^*$ assigned value $-1$ to all but a few small isolated components, which can never flip once formed. This effect was especially pronounced in networks with many high-weight vertices. 
The left two images of Figure~\ref{fig:square-combined} illustrate such a run: the blue region visible at $t=0$ (left) has completely disappeared in the final configuration $c^*$ (right).

For sufficiently large $s$, the blue region contracted during the early stages of the process, but the system then reached a stable configuration $c^*$ in which both opinions survived. The boundary of the blue region changed shape during this evolution. Vertices located at corners, surrounded by red neighbours in three quadrants, were highly unstable and flipped early, while vertices along flat edges with a more balanced neighbourhood were more likely to remain blue. The result after convergence was a rounded, approximately ball-shaped region of blue vertices. The left two images of Figure~\ref{fig:square-combined} show such an outcome, with the initially square set $c_0$ (left) evolving into a rounded persistent region in $c^*$ (right).
\begin{figure}[ht]
    \centering
    \includegraphics[width=0.7\textwidth]{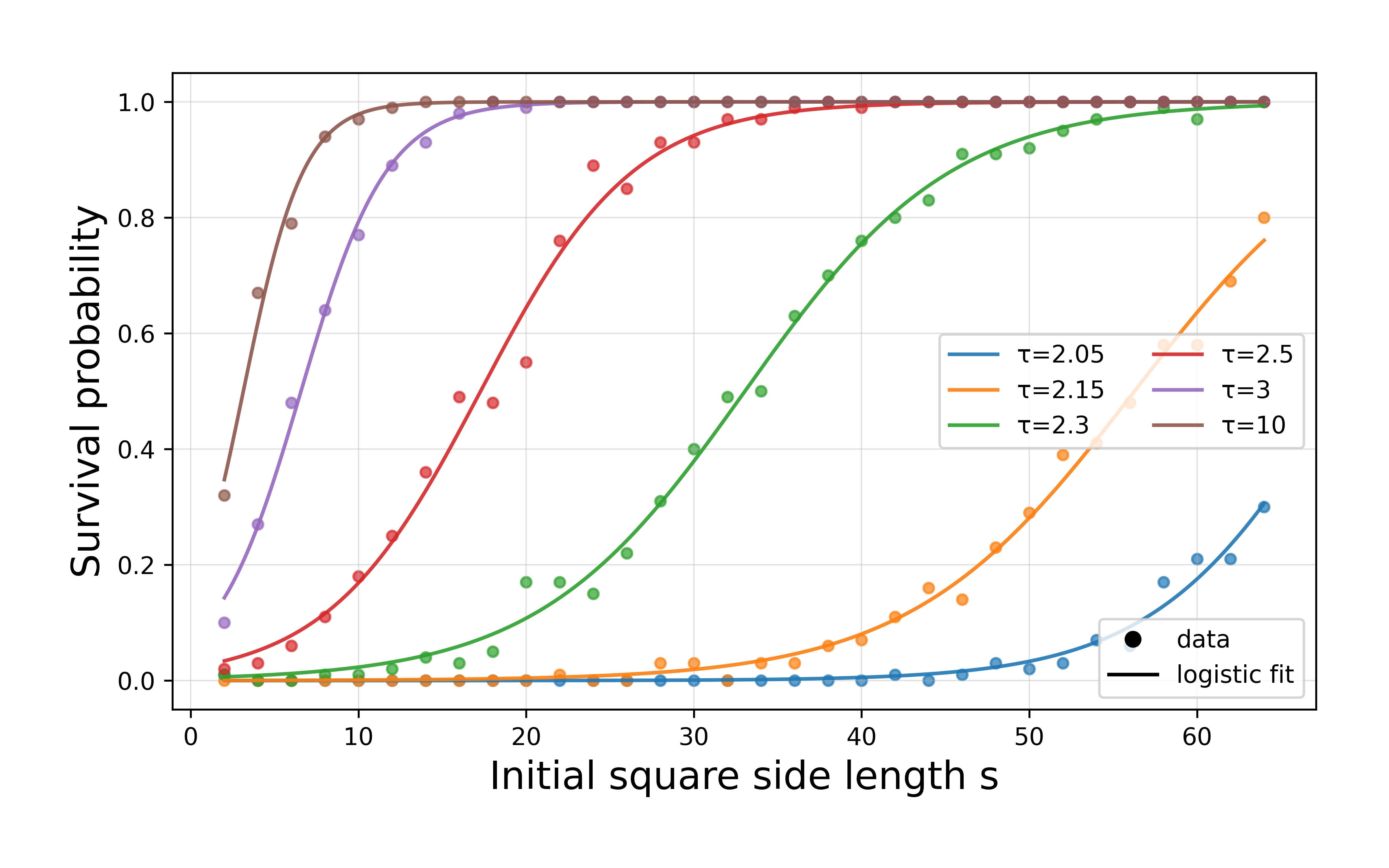}
    \vspace*{-10pt}
    \caption{Survival probability of a square initial configuration as a 
    function of its side length $s$, for various values of $\tau$. Dots 
    show simulation results (100 runs per data point), and solid curves 
    show logistic fits. The critical size decreases with $\tau$.}
    \label{fig:survival-curves}
\end{figure}
\noindent Systematic variation of $\tau$ showed a clear effect on the critical size needed for survival. As $\tau$ increased, survival became possible at smaller scales. Figure~\ref{fig:survival-curves} illustrates this behaviour: for each value of $\tau$, the survival probability of a square initial configuration rises sharply from near zero to near one as the side length $s$ grows, with the transition point shifting to smaller $s$ for larger $\tau$. In all cases, surviving regions gradually evolved towards rounded shapes. Heuristically, this fits with the observation that boundaries of high curvature tend to shrink: small balls disappear, while larger regions can contract into shapes with lower curvature everywhere. 
Our main theoretical result addresses this by showing (in a simplified mean-field setting) that the process indeed stabilizes in the limit of vanishing curvature.

\subsection{The Mean-Field Approximation}\label{sec:meanfield}
\subsubsection{Framework}
\noindent To analyse the Opinion Spreading process on GIRGs, we employ a mean-field approximation. The key idea is to replace the random, but fixed neighbourhood of each vertex by an independent sample from the same distribution that is re-drawn each round. This leads to a deterministic update operator acting on functions $f$ that describe the probability of holding opinion $+1$.

\begin{Definition}[Mean-Field Assumption]\label{def:mean-field}
Let $\mathrm{G} = (\mathrm{V},\mathrm{E})$ be a GIRG on $n$ vertices and large enough density parameter $k$. Each vertex $v \in V$ has a weight $w_v$ and a position $x_v \in \mathcal{X}\subseteq \mathbb{R}^d$. The \emph{Mean-Field Assumption} posits that, conditioned on the type $(w_v,x_v)$ of a vertex $v$, the opinions of its neighbours are independent. Each neighbour holds opinion $+1$ with probability $f(w_u,x_u)$, where $f : \mathcal{D}\times \mathcal{X} \;\longrightarrow\; [0,1]$ is a function describing the current distribution of opinions. Moreover, the number and type of neighbours are redrawn for each update.

This assumption neglects local spatial dependencies induced by geometry, but retains heterogeneity in vertex types and spatially varying connection densities.
\end{Definition}
\noindent This assumption induces a natural one-step evolution rule: given the current distribution $f$, we can compute the updated distribution of opinions in the next round by evaluating the probability that a vertex adopts opinion $+1$ under the mean-field model. Recall that we will study the case for sufficiently large average degree, which is controlled by the parameter $k$. 

\begin{Definition}[Mean-Field Update Operator]\label{def:update-operator}
Let $f:\mathcal{D}\times\mathcal{X}\to[0,1]$ be a function and let $\mathbbm{1}\big((w,x),(w',x')\big)$ denote the indicator of the event that vertices of types $(w,x)$ and $(w',x')$ are adjacent. For a vertex $v$ of type $(w,x)$ define
\begin{align*}
\lambda(w) &:= \mathbb{E}[|N(v)|] = \frac{\pi^{d/2}}{\Gamma\left(\frac{d}{2}+1\right)} k \, w_v \, \left( 1 + \frac{1}{\tau-2}\right), \\
\lambda_{1,f}(w,x) &:= \int f(w',x') \, \mathbbm{1}\big((w,x),(w',x')\big)\, d\eta(w',x'), \\
\lambda_{-1,f}(w,x) &:= \lambda(w)-\lambda_{1,f}(w,x), \\ 
\mu_f(w,x) &:= \lambda_{1,f}(w,x)-\lambda_{-1,f}(w,x),
\end{align*}
where $\eta$ denotes the density of existence of a vertex of type $(w,x)$. Note that $\lambda(w)$ is simply the expected number of neighbours and does not depend on $f$ or~$x$. 

Under the mean-field assumption, the numbers of neighbours with opinions $\pm 1$ are independent Poisson random variables with these means. We define the \emph{mean-field update operator} $\hat{\mathcal{T}}$ by
\begin{align*}
    (\hat{\mathcal{T}}f)(w,x) := \Pr\!\left[ \mathrm{Po}\big(\lambda_{1,f}(w,x)\big) - \mathrm{Po}\big(\lambda_{-1,f}(w,x)\big) \geq 0 \right].
\end{align*}
where we adopt the convention that ties count as $+1$; this differs from the sequential rule, but the exact tie event has vanishing probability for large degrees and can thus be ignored.
Moreover, for sufficiently large $k$ we can use a Gaussian approximation of the above Skellam distribution via the central limit theorem,
\begin{align}
    (\hat{\mathcal{T}}f)(w,x) \;\approx\; (\mathcal{T}f)(w,x) \; := \;  \Phi\!\Big(\tfrac{\mu_f(w,x)}{\sqrt{2\,\lambda(w)}}\Big), \label{eq:1}
\end{align}
where $\Phi$ denotes the standard normal distribution function.
\end{Definition}
\noindent In the following we will work with the operator $\mathcal{T}$ instead of $\hat{\mathcal{T}}$, meaning that our results assume sufficiently large $k$.
\begin{Definition}[Advantage]\label{def:advantage}
For a vertex $v$ of type $(w,x)$ and a function $f$, the quantity $\mu_f(w_v,x_v)$ defined in Definition~\ref{def:update-operator} is called the \emph{advantage} of $v$ under $f$. It measures the expected difference between the numbers of neighbours of $v$ holding opinions $+1$ and $-1$. We use the notation $\mu^\mathcal{R}_f(w_v,x_v)$ to denote the advantage of $v$ in the region $\mathcal{R}$.
\end{Definition}

\noindent Next we define what it means for an opinion to survive in the mean-field.
\begin{Definition}[Survival in the Mean-Field]\label{def:mf-robust-survival}
Let $f_0$ be some initial configuration. Define $f^t = \mathcal{T}^t f_0$ and $f^\star :=  \lim_{t\to\infty} f_t$. We say that both opinions $t$-\emph{survive} in the mean-field if there exist a constant $\varepsilon$, such that for all $w\in \mathcal D$, there are positions $x_{-1}$, $x_{+1}$ that satisfy
\begin{align*}
  f_t(w,x_{+1})\ge\tfrac12+\varepsilon \quad\text{and}\quad f_t(w,x_{-1})\le\tfrac12-\varepsilon.
\end{align*}
If this holds for $f^\star$, we will drop the $t$ in $t$-survival and simply say that both opinions survive.
\end{Definition}

\noindent Although vertices in a GIRG are embedded in the full cube $\mathcal X \subset \mathbb R^d$, our analysis will for the most part only depend on their signed distance from the boundary of a half-space. We therefore write $z := x[0]\in\mathbb R$ for the first coordinate of $x\in\mathcal X$, and restrict attention to functions of the form $f:\mathcal D\times\mathbb R \to [0,1];\ (w,z)\mapsto f(w,z)$. 

To estimate the advantage, we will often have to compare different regions of the neighbourhood of a vertex. The following definition will be useful for that.

\begin{Definition}[Complement] \label{def:complement-vertex}
    Let $x \in \mathcal{X}$ be a position with coordinates $(x[0],x[1],\dots,x[d-1])$. 
    We define its \emph{complement} $x^\dagger$ as the reflection of $x$ across the hyperplane $\{z=0\}$. 
    Formally,
    $x^\dagger[0] := -\,x[0]$ and 
$x^\dagger[i] := x[i]$ for all $i=1,\dots,d-1$.
\end{Definition}
\noindent The next definition will be at the heart of the proof. We will show later that any valid subsolution of $\mathcal T$ is a bound for the limiting survival function $f^*$ on half-spaces. 

\begin{Definition}[Valid Function]\label{def:f_valid}
A function $f:\mathcal D\times\mathbb R \to [0,1]$ is called \emph{valid} if it satisfies Conditions 1--3 below, and we call $f$ a \emph{valid subsolution of $\mathcal{T}$} if additionally Condition 4 holds.
\begin{enumerate}[label=\arabic*., ref=\thedefinition.\arabic*]
    \item \textbf{Symmetry:} $f(w,z)+f(w,-z)=1$ for all $(w,z)$; in particular $f(w,0)=\tfrac{1}{2}$.\label{def:f_valid_sym}
    \item \textbf{Monotonicity in $z$:} For each fixed $w$, the map $z\mapsto f(w,z)$ is monotone increasing. \label{def:f_valid_mon_z}
    \item \textbf{Monotonicity in $w$:} For each fixed $z\ge k^{1/d}$, the map $w\mapsto f(w,z)$ is monotone increasing on $[1,z^d/k]$. \label{def:f_valid_mon_w}
    \item \textbf{Subsolution:} $f(w,z)\le (\mathcal T f)(w,z)$ for all $(w,z)$ with $z\geq 0$. \label{def:f_valid_sub}
\end{enumerate}
\end{Definition}

\noindent Note that symmetry Condition~\ref{def:f_valid_sym} together with the subsolution Condition~\ref{def:f_valid_sub} implies the reverse inequality $f(w,z)\ge (\mathcal T f)(w,z)$ for negative $z \le 0$, as the update operator on $f$ is also symmetric.

\subsection{The Survival Theorem for Half-spaces} 
\noindent Looking at half-spaces is motivated by our experimental observation. Areas of high local curvature tend to flatten over time, either vanishing fully, or reaching stability once the curvature becomes negligible. In the limit, all shapes of negligible curvature locally look identical to half-spaces, which makes them a well-suited proxy for more complicated shapes.

We call the initialization $f_0(w,z) := \mathbbm{1}\{z\geq 0\}$, corresponding to the opinion configuration $c_0$ that assigns $-1$ to all vertices to the left of the half-space and $+1$ to the vertices on the right, the half-space initialization. Our main result of this section is to prove its mean-field survival.

\begin{restatable}[Survival Theorem]{Theorem}{survival}\label{thm:main}
For the half-space initialization $f_0$,
both opinions survive in the mean-field.
\end{restatable}
\noindent To prove Theorem~\ref{thm:main}, we first collect some easy facts regarding $\mathcal T$.

\begin{Lemma}[Continuity of $\mathcal{T}$]
\label{lem:T-con}
Let \(f:\mathcal D\times\mathbb R^d\to[0,1]\) and fix \(W\in\mathcal D\). Then the map
\begin{align*}
x\mapsto (\mathcal T f)(W,x)
\end{align*}
is continuous.
\end{Lemma}

\begin{proof}
Since $\Phi$ is continuous and \(\lambda(W)\) does not depend on \(x\), it is enough to show that
\[
x\mapsto \mu_f(W,x)
\]
is continuous. Let \(x_n\to x\). By definition,
\begin{align*}
    \mu_f(W,x_n)-\mu_f(W,x)
    &=
    \int_{\mathcal D}\int_{\mathbb R^d} \bigl(2f(w,y)-1\bigr)\\
    &\hspace{20pt}\cdot
    \left(
    \mathbbm{1}\bigl(\|x_n-y\|^d\le kWw\bigr)
    -
    \mathbbm{1}\bigl(\|x-y\|^d\le kWw\bigr)
    \right)
    dy\,\rho(w)\,dw.
\end{align*}
For fixed \(w\) and \(y\), the two indicators converge pointwise as \(n\to\infty\), except possibly on the boundary
\[
\|x-y\|^d=kWw.
\]
For each fixed \(w\), this boundary is a sphere in \(\mathbb R^d\). The integrand thus converges almost everywhere. Moreover, since \(0\le f\le 1\), the integrand is bounded in absolute value by
\[
2\,\mathbbm{1}\bigl(\|x-y\|^d\le kWw+1\bigr)
\]
for all sufficiently large \(n\). This is integrable, because
\[
\int_{\mathcal D}\int_{\mathbb R^d}\mathbbm{1}\bigl(\|x-y\|^d\le kWw+1\bigr)\,dy\,\rho(w)\,dw
=
\Theta(1)\int_{\mathcal D}(kWw+1)\rho(w)\,dw<\infty,
\]
where the last inequality uses \(\tau>2\). By dominated convergence,
\[
\mu_f(W,x_n)\to \mu_f(W,x).
\]
Applying continuity of \(\Phi\) yields
\[
(\mathcal T f)(W,x_n)\to (\mathcal T f)(W,x).
\]
Hence \(x\mapsto (\mathcal T f)(W,x)\) is continuous.
\end{proof}

\begin{Lemma}[Symmetry Preservation of $\mathcal{T}$]\label{lem:T-sym}
If $f: D\times\mathbb R \to [0,1]$ is symmetric, then $(\mathcal{T} f)$ also is.
\end{Lemma}
\begin{proof}
    By symmetry of $\Phi$
    \begin{align*}
        (\mathcal{T}f)(w,x)  =   \Phi\!\left(\frac{\mu_f(w,x)}{\sqrt{2\,\lambda(w)}}\right) = 1- \Phi\!\left(\frac{-\mu_f(w,x)}{\sqrt{2\,\lambda(w)}}\right).
    \end{align*}
    Since $f$ is symmetric, $\mu_f(w,x)$ also is. Thus
    \begin{align*}
        (\mathcal{T}f)(w,x) = 1- \Phi\!\left(\frac{\mu_f(w,-x)}{\sqrt{2\,\lambda(w)}}\right) = 1- (\mathcal{T}f)(w,-x),
    \end{align*}
    which proves the statement.
\end{proof}

\begin{Lemma}[Monotonicity of $\mathcal T$]\label{lem:T-monotone}
Define $f,g: D\times\mathbb R \to [0,1]$, then these two monotonicity statements hold for $\mathcal{T}$:
\begin{enumerate}[label=\roman*., ref=\roman*]
\item  If $f$ and $g$  are symmetric and satisfy $f(w, z)\le g(w, z)$ for all $w$ and $z\ge 0$ then $(\mathcal T f)(w, z) \le (\mathcal T g)(w, z)$ for all $w$ and $z\ge 0$. \label{lem:T_mon_1}
\item If $f$ and $g$ satisfy $f(w, z) \le g(w, z)$ everywhere, then also $(\mathcal T f)(w, z) \le (\mathcal T g)(w, z)$ everywhere. \label{lem:T_mon_2}
\end{enumerate}

\end{Lemma}
\begin{proof}
    We prove the two claims separately. By definition of $\mathcal T$ it suffices to show that $\mu_f(w,x) \le \mu_g(w,x)$ for all $w$ and $z\ge 0$. We write
    \begin{align*}
\mu_g(w,x) - \mu_f(w,x) &= \int 2(g(w',x')-f(w',x')) \, \mathbbm{1}\big((w,x),(w',x')\big)\, d\eta(w',x') \\
&= \int_{z'\ge0} 2(g(w',x')-f(w',x')) \, \mathbbm{1}\big((w,x),(w',x')\big)\, d\eta(w',x') \\
&\hspace{10pt}-\int_{z'\le0} 2(g(w',x')-f(w',x')) \, \mathbbm{1}\big((w,x),(w',x')\big)\, d\eta(w',x').
\end{align*}
By symmetry,
\begin{align*}
    \mu_g(w,x) - \mu_f(w,x) &= \int_{z'\ge0} 2(g(w',x')-f(w',x')) \, \mathbbm{1}\big((w,x),(w',x')\big)\, d\eta(w',x') \\
&\hspace{10pt}+\int_{z'\ge0} 2(f(w',x')-g(w',x')) \, \mathbbm{1}\big((w,x),(w',x'^\dagger)\big)\, d\eta(w',x'^\dagger).
\end{align*}
Since $z\ge 0$, $\lVert x-x'\rVert \le \lVert x-x'^\dagger\rVert $ for all $x'$ with $z'\ge 0$, so in particular,
\begin{align*}
    \mathbbm{1}\big((w,x),(w',x'^\dagger)\big) =1  \implies \mathbbm{1}\big((w,x),(w',x')\big) =1
\end{align*}
for all $w$. Therefore
\begin{align*}
    \mu_g(w,x) - \mu_f(w,x) &\ge \int_{z'\ge0} 2(g(w',x')-f(w',x')) \, \mathbbm{1}\big((w,x),(w',x')\big)\, d\eta(w',x') \\
&\hspace{10pt}+\int_{z'\ge0} 2(f(w',x')-g(w',x')) \, \mathbbm{1}\big((w,x),(w',x')\big)\, d\eta(w',x'^\dagger) \\
&\ge 0. 
\end{align*}
This proves claim~\ref{lem:T_mon_1}. Claim~\ref{lem:T_mon_2}, follows immediately from the definition
\begin{align*}
\mu_g(w,x) - \mu_f(w,x) &= \int 2(g(w',x')-f(w',x')) \, \mathbbm{1}\big((w,x),(w',x')\big)\, d\eta(w',x') \\
&> 0,
\end{align*}
since $g(w',x')-f(w',x') >0$ everywhere.
\end{proof}

\noindent Applying the symmetry and monotonicity recursively, we obtain the following result.

\begin{Theorem}[Comparison Principle for $\mathcal{T}$]\label{thm:valid-lb}
Let $f_0$ be the half-space initialization and $f_t := \mathcal T^t f_0$. Furthermore let $f$ be a valid subsolution of $\mathcal T$ with $f(w,z)\leq f_0(w,z)$ on $z \in [0, \infty)$. Then  $f(w,z)\le f_t(w,z)$ for all $t$, $w$ and $z\ge0$. By symmetry also $f(w,-z)\ge f_t(w,-z)$.
\end{Theorem}
\begin{proof}
By Condition~\ref{def:f_valid_sym}, $f \leq \mathcal{T} f$. By Lemmas~\ref{lem:T-monotone} and~\ref{lem:T-sym}, 
\begin{align*}
    f \leq \mathcal{T} f \leq \mathcal{T}^t f \leq \mathcal{T}^t f_0 = f_t. \qedhere
\end{align*}
\end{proof}

\noindent With these facts, the aim of the remaining proof is to construct a valid subsolution $f$ of $\mathcal{T}$. The most crucial step in this construction will be  obtaining a pointwise lower bound on the \emph{advantage} of a vertex $v$ under $f$. This bound will then allow us to define $f$ via the right-hand side of the defining equation of the update Operator~\eqref{eq:1}.

\subsubsection{Geometric Advantage Partitioning}
\noindent We partition the ambient space into two regions and compute the contribution to the advantage from each separately.
\begin{Definition}[Partitioning]
Let $x \in \mathcal{X}$ and let $x^\dagger$ denote its complement (Definition~\ref{def:complement-vertex}). Let 
\begin{align*}
    \mathcal{B}_{\geq 0} &:= \{\,x \in \mathcal{X} : z \geq r_I+z_v/2\,\} \cap I(v),
\end{align*}
and let
\begin{align*}
\mathcal{B} &:= \mathcal{B}_{\geq 0} \cup  \mathcal{B}_{\geq 0}^{\dagger},\\
\mathcal{R} &:= \mathcal{X} \setminus \mathcal{B}.
\end{align*}
Then $\mathcal{B}$ and $\mathcal{R}$ form a partition of $\mathcal{X}$.  A visualization of this partition can be found in Figure~\ref{fig:spacePartitioning}.
\begin{figure}[ht] 
\centering
\begin{tikzpicture}[scale=0.3, transform shape, every node/.style={transform shape=false}]
  \def\R{5} 

  \fill[darkRed!70] (-9,-8) -- (-9,8) -- (9,8) -- (9,-8) -- (-9,-8) -- cycle;

\begin{scope}
  \clip (3,0) circle[radius=\R];
  \fill[darkBlue] (-20,-20) rectangle (20,20);
  \begin{scope}
    \fill[darkRed!70] (-2,-5) -- (-2,5) -- (6.5,5) -- (6.5,-5) -- (-2,-5) -- cycle;
  \end{scope}
\end{scope}

\begin{scope}
  \clip (-3,0) circle[radius=\R];
  \fill[darkBlue] (-20,-20) rectangle (20,20);
  \begin{scope}
    \fill[darkRed!70] (2,-5) -- (2,5) -- (-6.5,5) -- (-6.5,-5) -- (2,-5) -- cycle;
  \end{scope}
\end{scope}
\draw[darkRed!70, line width=1.2pt] (-3,0) circle[radius=\R];

  \draw[thick] (0,-8) -- (0,8);


  \draw[dotted, thick] (3,0) -- (5.5,4.3301);
  \draw[dotted, thick] (0,0) -- (3,0);
  \node[above] at (5,0.8) {$r_{I}$};
  \node[above] at (1.5,0) {$z_v$};
  \draw[dotted, thick] (0,-6) -- (6.5,-6);
  \draw (6.5,-8) -- (6.5,3.5707);
  \node[below] at (3.25,-6) {$r_{I}+z_v/2$};

  \draw[dotted, thick] (0,6) -- (8,6);
  \draw (8,8) -- (8,0);
  \node[above] at (4,6) {$r_{I}+z_v$};
  
  \draw[black, thick] (3,0) circle[radius=\R];

  \fill[black] (3,0) circle[radius=0.2];
  \node[anchor=north] at (3,0) {$v$};
  \fill[black] (-3,0) circle[radius=0.2];
  \node[anchor=north] at (-3,0) {$v^\dagger$};
\end{tikzpicture}

\caption{The two dimensional space $\mathcal{X}$ partitioned into the two regions $\mathcal{R}$ and $\mathcal{B}$.}
\label{fig:spacePartitioning}
\end{figure}
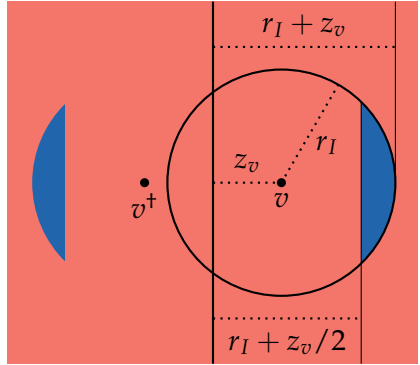
\end{Definition}
\noindent We start by bounding the advantage of the red region.
\begin{Lemma} \label{lem:advantage_red}
    The \emph{advantage} $\mu_f^\mathcal{R}(w_v, x_v)$ of the red region is non-negative for any valid function $f$ and  any vertex $v$ with $0\leq z_v$.
\end{Lemma}
\begin{proof}
    By symmetry of $\mathcal{B}$ across the hyperplane, $\mathcal{R}$ must also be. Unfolding the advantage definition
    \begin{align*}
     \mu_f^\mathcal{R}(w_v, x_v) &= \int_{\mathcal R}(2f(w,z)-1)\,\mathbbm{1}\big((w_v,x_v),(w,x)\big)\,d\eta(w,x)\\
     &= \int_{\mathcal R_{\ge 0}}(2f(w,z)-1)\,\mathbbm{1}\big((w_v,x_v),(w,x)\big)\,d\eta(w,x) \\
     & \hspace{10pt}+\int_{\mathcal R_{\ge 0}^\dagger}(2f(w,z)-1)\,\mathbbm{1}\big((w_v,x_v),(w,x)\big)\,d\eta(w,x) \\
     &= \int_{\mathcal R_{\ge 0}}(2f(w,z)-1)\,\mathbbm{1}\big((w_v,x_v),(w,x)\big)\,d\eta(w,x) \\
     & \hspace{10pt}-\int_{\mathcal R_{\ge 0}}(2f(w,z)-1)\,\mathbbm{1}\big((w_v,x_v),(w,x^\dagger)\big)\,d\eta(w,x).
\end{align*}
Since $z_v\ge 0$, $\lVert x_v-x\rVert \le \lVert x_v-x^\dagger\rVert $  for all $x$ with $z\ge 0$, so in particular
\begin{align*}
    \mathbbm{1}\big((w_v,x_v),(w,x^\dagger)\big) =1  \implies \mathbbm{1}\big((w_v,x_v),(w,x)\big) =1
\end{align*}
for all $w$. Therefore 
\begin{align*}
    \mu_f^\mathcal{R}(w_v, x_v) \ge 0. \qedhere
\end{align*}
\end{proof}

\noindent Next we bound the advantage of the blue region.
\begin{Lemma} \label{lem:advantage_blue}
    The advantage $\mu_f^\mathcal{B}(w_v, z_v)$ of the blue region is lower bounded by
    \begin{align*}
    \mu_f^\mathcal{B}(w_v, x_v)\geq \underline{\mu}_f^\mathcal{B}(w_v, z_v)
    \end{align*}
    for any valid function $f$ and any vertex $v$ with $0\leq z_v\leq r_{I(v)}$, where
    \begin{align*}
        \underline{\mu}_f^\mathcal{B}(w_v, z_v) = \Theta(1) \, z_v^{(d+1)/2} \, r_I^{(d-1)/2} \left(1- \left(1+\frac{{z_v}}{2\,r_I}\right)^{d(1-\tau)}\right)\left(f(1,r_I)-\frac{1}{2}\right)
    \end{align*}
\end{Lemma}
\begin{proof}
Write $r_I:=r_{I(v)}$ for brevity. The volume of the blue space on the right side of the hyperplane is a cut ball. The intersection of a hyperplane and a ball is either empty, or again a ball, but of diminished dimension. Thus, by construction, $I(v) \cap \{u: z_u = r_I + z_v/2 \}$ is a $(d-1)$-dimensional ball. Applying the Pythagorean theorem, the radius $r$ of this ball can be bounded by
\begin{align*}
    r &= \left(r_I^2-\left(r_I-\frac{z_v}{2}\right)^{2}\right)^{1/2} = \left(r_I \, z_v-\frac{z_v^2}{4}\right)^{1/2} \\
    &\ge \left(\frac{r_I \,z_v}{2}\right)^{1/2},
\end{align*}
where the last step exploits the assumption $r_I \ge z_v$. Using this, we can lower bound the volume of the region $\mathcal{B}_{\geq0}$, by the volume of the $d$-cone $C_{r, z_v/2}^{d}$ with  base $B_r^{d-1}$ and height $z_v/2$. Applying the volume definitions for the cone and the ball
\begin{align*}
     \mathrm{Vol}\left(\mathcal{B}_{\geq0}\right) &\ge \mathrm{Vol}\left(C_{r, z_v/2}^{d}\right) = \frac{z_v}{2\,d}\mathrm{Vol}\left({B_r^{d-1}}\right) = \frac{z_v^{(d+1)/2}}{2^{(d+1)/2}\,d} \frac{\pi^{(d-1)/2}}{\Gamma\left(\frac{d-1}{2}+1\right)} r_I^{(d-1)/2} \\
     &\ge \Theta(1) \, z_v^{(d+1)/2} \, r_I^{(d-1)/2},
\end{align*}
where $\Gamma $ is the Gamma function, a function which extends the factorial operator to the reels.

Now that we have gotten a handle for the volume of the blue region, we can shift our attention to the advantage. Unfolding the advantage definition, splitting the region into its reflexive halves and applying symmetry (Condition~\ref{def:f_valid_sym})
\begin{align*}
     \mu_f^\mathcal{B}(w_v, x_v) &= \int_{\mathcal B}(2f(w,z)-1)\,\mathbbm{1}\big((w_v,x_v),(w,x)\big)\,d\eta(w,x)\\
     &= \int_{\mathcal B_{\ge 0}}(2f(w,z)-1)\,\mathbbm{1}\big((w_v,x_v),(w,x)\big)\,d\eta(w,x) \\
     & \hspace{10pt}+\int_{\mathcal B_{\ge 0}^\dagger}(2f(w,z)-1)\,\mathbbm{1}\big((w_v,x_v),(w,x)\big)\,d\eta(w,x) \\
     &= \int_{\mathcal B_{\ge 0}}(2f(w,z)-1)\,\mathbbm{1}\big((w_v,x_v),(w,x)\big)\,d\eta(w,x) \\
     & \hspace{10pt}-\int_{\mathcal B_{\ge 0}}(2f(w,z)-1)\,\mathbbm{1}\big((w_v,x_v),(w,x^\dagger)\big)\,d\eta(w,x).
\end{align*}
Any vertex in $\mathcal{B}_{\ge 0}$ is adjacent to $v$, since $\mathcal{B}_{\ge 0} \subset I(v)$. Furthermore, all vertices $u \in \mathcal{B}_{\ge 0}^\dagger$ have distance at least $r_I + 3z_v/2$ from $v$. This lets us solve for $w_u$ in the GIRG edge criterion, which gives the necessary condition $w_u \ge (r_I + 3z_v/2)^d/(kw_v)$ for all neighbours of $v$ in $\mathcal{B}_{\ge 0}^\dagger$. Using this observation, we can lower bound $\mu_f^\mathcal{B}(w_v, z_v)$ by
\begin{align*}
    \mu_f^\mathcal{B}(w_v, x_v) &\ge \int_{\mathcal B_{\ge 0}} \int_{1}^\infty (2f(w,z)-1)\, \rho(w)\,dw\,d\eta(x) \\
     & -\int_{\mathcal B_{\ge 0}}\int_{\frac{\left(r_I + 3 z_v/2\right)^d}{k \,w_v}}^\infty (2f(w,z)-1)\, \rho(w)\,dw\,d\eta(x) \\
     &\ge \int_{\mathcal B_{\ge 0}} \int_{1}^\infty (2f(w,z)-1)\, \rho(w)\,dw\,d\eta(x) \\
     & -\int_{\mathcal B_{\ge 0}}\int_{\frac{\left(r_I + z_v/2\right)^d}{k \,w_v}}^\infty (2f(w,z)-1)\, \rho(w)\,dw\,d\eta(x) \\
     &\ge \int_{\mathcal B_{\ge 0}} \int_{1}^{\frac{\left(r_I + z_v/2\right)^d}{k \,w_v}} (2f(w,z)-1)\, \rho(w)\,dw\,d\eta(x).
\end{align*}
By construction of $\mathcal B_{\ge 0}$, all vertices have a $z$-component of at least $r_I + z_v/2$. By monotonicity in $z$ (Condition~\ref{def:f_valid_mon_z}), this lets us pointwise lower bounds  $f(w, z)$ over the integration range of the outer integral
\begin{align*}
    \mu_f^\mathcal{B}(w_v, x_v) \ge \int_{\mathcal B_{\ge 0}} \int_{1}^{\frac{\left(r_I + z_v/2\right)^d}{k \,w_v}} (2f(w,r_I+z_v/2)-1)\, \rho(w)\,dw\,d\eta(x).
\end{align*}
By monotonicity of $f(w, z)$ in $w$ on $[1, z^d/k]$ (Condition~\ref{def:f_valid_mon_w}), $f$ can be pointwise lower bounded over the integration range of the inner integral
\begin{align*}
    \mu_f^\mathcal{B}(w_v, x_v) &\ge \int_{\mathcal B_{\ge 0}} \int_{1}^{\frac{\left(r_I + z_v/2\right)^d}{k \,w_v}} (2f(1,r_I+z_v/2)-1)\, \rho(w)\,dw\,d\eta(x) \\
    & \ge \Theta(1)  \left(f(1,r_I)-\frac{1}{2}\right)\int_{\mathcal B_{\ge 0}} \int_{1}^{\frac{\left(r_I + z_v/2\right)^d}{k \,w_v}} \rho(w)\,dw\,d\eta(x).
\end{align*}
Notice that the integrand of the outer interval is independent of $x$, this lets us bound the expression by applying the previously derived bound for the volume of $\mathcal{B}_{\ge 0}$:
\begin{align*}
    \mu_f^\mathcal{B}(w_v, x_v) &\ge  \Theta(1) \, z_v^{(d+1)/2} \, r_I^{(d-1)/2}\left(f(1,r_I)-\frac{1}{2}\right) \int_{1}^{\frac{\left(r_I + z_v/2\right)^d}{k \,w_v}}\rho(w) \, dw.
\end{align*}
Finally solving the inner integral
\begin{align*}
    \mu_f^\mathcal{B}(w_v, x_v) &\ge\Theta(1) \, z_v^{(d+1)/2} \, r_I^{(d-1)/2}\left(f(1,r_I)-\frac{1}{2}\right) \int_{1}^{\frac{\left(r_I + z_v/2\right)^d}{k \,w_v}} (1-\tau) \, w^{-\tau} \, dw \\
    &\ge \Theta(1) \, z_v^{(d+1)/2} \, r_I^{(d-1)/2}\left(f(1,r_I)-\frac{1}{2}\right) \left[-w^{1-\tau}\right]_1^{\frac{\left(r_I + z_v/2\right)^d}{k \,w_v}} \\
    &\ge  \Theta(1) \, z_v^{(d+1)/2} \, r_I^{(d-1)/2} \left(1- \left(\frac{{2\,r_I+z_v}}{2\, r_I}\right)^{d(1-\tau)}\right)\left(f(1,r_I)-\frac{1}{2}\right) \\
    &\ge  \Theta(1) \, z_v^{(d+1)/2} \, r_I^{(d-1)/2} \left(1- \left(1+\frac{{z_v}}{2\,r_I}\right)^{d(1-\tau)}\right)\left(f(1,r_I)-\frac{1}{2}\right),
\end{align*}
we concludes the proof by defining $\underline{\mu}_f^\mathcal{B}(w_v, z_v)$ to be the right hand side.
\end{proof}
\subsubsection{A Valid Subsolution}
\noindent Before finally defining a valid subsolution of $\mathcal{T}$, we need one more lemma about solutions to the standard normal distribution function $\Phi$ that appears in the definition of $\mathcal{T}$~\eqref{eq:1}. 

\begin{Lemma}\label{lem:f_fixpoint}
    For fixed $y > \sqrt{\pi}$, the equation
    \begin{align*}
        \delta = \Phi\!\left(y\,\left(\delta-\frac{1}{2}\right)\right)
    \end{align*}
    admits a solution $\delta^*= 1/2+\varepsilon$ for some $\varepsilon>0$.
\end{Lemma}
\begin{proof}
    Set
    \begin{align*}
        g(\delta) &=\Phi\!\left(y\,\left(\delta-\frac{1}{2}\right)\right) -\delta \qquad \text{for }\delta\in\bigl[\tfrac12,1\bigr].
    \end{align*}
    Since $\Phi$ is smooth, $g$ is continuous on $[1/2,1]$ and differentiable on $(1/2,1)$. A direct computation gives
    \begin{align*}
        g'(\delta)=\frac{y}{\sqrt{\pi}}\exp\left(-y^2\left(\delta-\frac12\right)^2\right)-1.
    \end{align*}
    Hence
    \begin{align*}
        g'\left(\tfrac12\right)=\frac{y}{\sqrt{\pi}}-1,
    \end{align*}
which is positive whenever $y>\sqrt{\pi}$. Thus there exists $\varepsilon>0$ such that $g(1/2+\varepsilon)>0$. On the other hand
\begin{align*}
    g(1)=\Phi\!\left(\frac{y}{2}\right)-1<0,
\end{align*}
because $\Phi(t)<1$ for every finite $t$. By the intermediate value theorem there exists
\begin{align*}
    \delta^*\in\left[\tfrac12+\varepsilon,\,1\right)\subset\left(\tfrac12,1\right)
\end{align*}
with $g(\delta^*)=0$, i.e.
\begin{align*}
    \delta^*=\Phi\!\left(y\left(\delta^*-\frac{1}{2}\right)\right),
\end{align*}
as required.
\end{proof}
\noindent This leads us to a function that is a valid subsolution of $\mathcal T$.
\begin{Lemma} \label{lem:valid_f}
    There exists a value $ \delta^* =1/2+\varepsilon$ for some $\varepsilon>0$ such that the function
    \begin{align*}
        f(w_v, z_v)= \begin{cases}
                        \delta^* &\text{ for $w_v=1$ and $z_v=r_{I(v)}$}\\\Phi\!\left(\frac{\underline{\mu}_f^\mathcal{B}\left(w_v, \min\{z_v, r_{I(v)}\}\right)}{\sqrt{ 2 \, \lambda(w_v)}}\right) &\text{ for $z_v\geq 0$} \\
                        1- f\left(w_v, -z_v\right) & \text{otherwise}
                    \end{cases} 
    \end{align*}
    is a valid subsolution of $\mathcal T$ for  some $k \in \Theta(1) $.
\end{Lemma}
\begin{proof}
We start by selecting $\delta^*$ as the solution to the equation
\begin{align*}
    f(1,k^{1/d}) &=\,\Phi\!\left(\frac{\underline{\mu}_f^\mathcal{B}(1, k^{1/d})}{\sqrt{ 2 \, \lambda(k^{1/d})}}\right) \\
        &= \,\Phi\!\left(\Theta(1) \, k^{1/2} \left(f(1,k^{1/d})-\frac{1}{2}\right)\!\right).
\end{align*}
Such a solution exists by Lemma~\ref{lem:f_fixpoint} whenever
\begin{align*}
        \Theta(1)\, k^{1/2} > \sqrt{\pi},
\end{align*}
where the $\Theta(1)$ hides constants in $d$ and $\tau$.
Rearranging yields, that this condition is satisfied whenever $k$ is a large enough constant in $d$ and $\tau$.

Next we go one by one trough the Conditions of Definition~\ref{def:f_valid}. The first Condition~\ref{def:f_valid_sym} follows from the construction of $f$ for all $z_v\not=0$. For $z_v=0$, $\underline{\mu}_f^\mathcal{B} = 0$ and the expression thus reduces to $f(w_v, 0) = \Phi(0) = 1/2$.
    
    The second Condition~\ref{def:f_valid_mon_z}, follows immediately from $\Phi$ and $\underline{\mu}_f^\mathcal{B}$ being monotone increasing in $z_v \ge 0$. Note that by choice of $\delta^*$, for vertices $v$ with $w_v =1$ and $z_v=r_{(v)}$, 
    \begin{align*}
        \delta^* = f(1, k^{1/d}) = \,\Phi\!\left(\frac{\underline{\mu}_f^\mathcal{B}(1, k^{1/d})}{\sqrt{ 2 \, \lambda(k^{1/d})}}\right) = \Phi\!\left(\frac{\underline{\mu}_f^\mathcal{B}\left(w_v, \min\{z_v, r_{I(v)}\}\right)}{\sqrt{ 2 \, \lambda(w_v)}}\right),
    \end{align*}
    monotonicity at this point is also guaranteed. By symmetry, monotonicity for $z_v < 0$ follows immediately.

    For the third Condition~\ref{def:f_valid_mon_w}, we can exploit the monotonicity of $\Phi$, to focus solely on $\underline{\mu}_f^\mathcal{B}$ and $\lambda$. The restriction $w_v \in [1, z_v^d/k]$ implies that $r_{I(v)} \leq z_v$, then
    \begin{align*}
        f(w_v, r_{I(v)}) &= \Phi\!\left(\frac{\underline{\mu}_f^\mathcal{B}\left(w_v, r_{I(v)}\right)}{\sqrt{ 2 \, \lambda(w_v)}}\right) \\
        &=\Theta(1) \, r_{I(v)}^{d/2}\left(\delta^* -\frac{1}{2} \right) \\
        &= \Theta(1) \, r_{I(v)}^{d/2} = \Theta(1) \, k \, w_v,
    \end{align*}
    which is clearly monotone increasing in $w_v$. Note again, that by choice of $\delta^*$, monotonicity at $w_v=1$ is also guaranteed.

    Since Conditions~\ref{def:f_valid_sym}--\ref{def:f_valid_mon_w} hold, $f$ is valid. We can show Condition~\ref{def:f_valid_sub} by applying Lemmas~\ref{lem:advantage_red} and \ref{lem:advantage_blue}:
    \begin{align*}
        f(w_v, z_v) =&\, \Phi\!\left(\frac{\underline{\mu}_f^\mathcal{B}(w_v, \min\{z_v, r_{I(v)}\})}{\sqrt{ 2 \, \lambda(w_v)}}\right) \\
        \leq&\, \Phi\!\left(\frac{\underline{\mu}_f^\mathcal{B}(w_v, z_v)}{\sqrt{ 2 \, \lambda(w_v)}}\right) \\
         \leq&\, \Phi\!\left(\frac{\mu_f(w_v, z_v)}{\sqrt{ 2 \, \lambda(w_v)}}\right) \\
         =&(\mathcal{T}f)(w_v, z_v).
    \end{align*}
    Thus $f$ is a valid subsolution of $\mathcal T$.
\end{proof}

\noindent We are now ready to prove Theorem~\ref{thm:main}.

\begin{proof}[Proof of Theorem~\ref{thm:main}]
Take $f$, $\delta^*$ as defined in Lemma~\ref{lem:valid_f}. Choose $k$ large enough such that both its requirement in Definition~\ref{def:mean-field} and its requirement in Lemma~\ref{lem:valid_f} are satisfied. Clearly $f(w,z) \leq f_0(w,z)$ on $z \in [0, \infty]$. Then by Theorem~\ref{thm:valid-lb}, $f(w,z) \leq f_t(w,z)$ for all $t$, $w$ and $z \in [0, \infty]$. Define $z_0(w) := k^{1/d} \, w^{1/d}$. It is easily verifiable that  $f(w, z) = f(w, z_0(w))\geq \delta^*$ for all $z\geq z_0(w)$, where 
\begin{align*}
    \delta^* = \tfrac{1}{2}+ \varepsilon.
\end{align*}
By symmetry also 
\begin{align*}
    f(w, -z) \le \tfrac{1}{2}- \varepsilon.
\end{align*}
Taking the limit $\lim_{t\to\infty} f_t$ yields the same bounds for $f^*$. Thus, both opinions survive in the mean-field.
\end{proof}

\subsection{From Half-Spaces to Balls}
\noindent In this section, we will convert our results from the previous section to large Euclidean balls. We will prove that a generalized version of Theorem~\ref{thm:main} is unobtainable in the limit, due to the lacking symmetry of the initial condition. However, $t$-survival can still be guaranteed for an unbounded number of mean-field updates. In other words, in the mean-field the boundary of the ball shrink by $o(1)$ in any constant time as the radius of the ball grows.

\subsubsection{Localization}
\noindent Remember our motivation to study the behaviour of the dynamic for the half-space interface: locally, the boundary of a large ball has negligible curvature and thus looks identical to a half-space. 

Let us start, by 
aligning the machinery we have developed for half-spaces to the new problem of Euclidean balls. Thus again, let $v$ be some vertex of type $(w_v, x_v)$ and let $B_r$ be a large ball of radius $r$. Previously $z_v$ represented the signed distance from $x_v$ to the half-space boundary. The equivalent concept in this setting, the signed distance of $x_v$ to the boundary of $B_r$, could have contributions from all different dimensions of $x_v$. But conveniently we can rotate the space in such a way that $v$ lies on the $z$-axis and that the ball has the half-space $z=0$ as tangent in direction of $z$. 
\begin{Lemma}[Local Coordinates] \label{lem:local_coordinates}
    Let $v$ be some vertex of type $(w_v, x_v)$ and let $B_r$ be a ball of radius $r$. Denote by $p$ the point on the boundary of $B_r$ closest to $x_v$ and $H_p$ as the hyperplane tangential to $B_r$ at $p$. Then there exists a distance preserving transformation $\mathrm{A}_v$ of $\mathcal{R}$, such that $A_v \cdot H_p = \{z=0\}$ and $A \cdot B_r $ is to the right of  $A \cdot H_p$.
\end{Lemma}
\begin{proof}
    We define $A_v = A_2 A_1$ as the combination of rotations and translations. Since $\lVert.\rVert_2$ is rotationally invariant, $A$ will be guaranteed to be distance preserving.

    Define $A_1$ to be the translation matrix that maps $p$ to the origin. Let $\vec{u}$ be the normalized normal vector of $H_p$ towards $B_r$, then by the transitivity of the rotation group $SO(d)$ on the sphere $S^{d-1}$, there exists a rotation $A_2 \in SO(d)$ with $A_2 \vec{u} = e_0$. Then for any $x \in \mathcal{X}$,
    \begin{align*}
        x \in A_1H_p \Leftrightarrow \vec{u}\cdot A_1 x = 0 \Leftrightarrow (A_2\vec{u})\cdot (A_2 A_1 x) =0 \Leftrightarrow e_0 (A_2A_1x) = 0 \Leftrightarrow A_2 A_1x [0] = 0.
    \end{align*}
    Thus $A_v H_p =A_2 A_1 H_p =  \{z=0\}$.
\end{proof}

\subsubsection{Non-Stability of Balls in the Mean-Field}
\noindent We define the natural ball initialization
\begin{align*}
    g_0(w, x) := \begin{cases}
        1 & \text{if } x \in B_r \\
        0 & \text{otherwise}
    \end{cases} 
\end{align*}
its evolution under the mean-field $g_{t+1}(w, x) := (\mathcal T g_t)(w,x)$  and its limiting distribution $g^* := \lim_{t\rightarrow \infty} {g_{t}}$. As for the half-space initialization, we are interested in the behaviour of its limiting distribution. Unfortunately, $g^*$ is, under some mild assumption, much less interesting than $f^*$.
\begin{Theorem} \label{thm:non_convergence}
    Fix $B_r$ and a sufficiently large weight cut-off $W \in \Theta(n)$. Consider the mean-field model, in which $\mathcal D$ is truncated at $W$. Then $g^* < 1/2$ everywhere.
\end{Theorem}
\noindent The only additional assumption needed for Theorem~\ref{thm:non_convergence} is the imposition of a global weight cut-off $W$ onto the mean-field model. This is motivated by the underlying finite GIRG. Indeed, in a GIRG on $n$ vertices with weight tail exponent $\tau>2$, the maximum weight is of
order
\begin{align*}
    w \in\Theta\left(n^{\frac{1}{\tau-1}}\right).
\end{align*}
with high probability. On this event, truncating the weight distribution at $W$ does not remove any vertex, hence it does not change the realized graph (when keeping the positions fixed) and therefore does not change the dynamics. Consequently, the truncated mean-field model
can be viewed as the natural mean-field analogue of the finite GIRG on the high-probability event
that no vertex exceeds weight $W$. In the limit of $n\rightarrow\infty$, such a weight cut-off does not significantly change the results established in Section~\ref{sec:definitions}. It also does not falsify any results established in Section~\ref{sec:meanfield}, as these rely only on the properties of vertices with small weight.
\begin{proof}[Proof of Theorem~\ref{thm:non_convergence}]
Let $B_t = \{x : \exists w,\,  g_t(w, x) \ge 1/2\}$ be the set of all positions that violate the statement of the lemma at time $t$. We will show that the set $B_\infty$ is empty.

    We start by showing that for all $t$, all points in $B_t$ must have come from the original ball $B_r$. Let $x \not \in B_r$ and let $H$ be a separating hyperplane that is tangential to $B_r$ at the point that is closest to $x$. Transform this into local coordinates using Lemma~\ref{lem:local_coordinates} such that $B_r$ ends up on the right side of the hyperplane. Then by construction $g_0 \le f_0$ everywhere, for the half-space initialization $f_0$ on these local coordinates. By  Lemma~\ref{lem:T-monotone}.\ref{lem:T_mon_2}, $g_t \le f_t$ for all $t$, so in particular, $g^* \le f^*$. Let $z = x[0]$ in local coordinates. Since $H$ separated $x$ from $B_r$,  $z<0$. By Lemma~\ref{lem:valid_f}, $f^*(\cdot, z)$ is upper bounded in this regime by a value strictly below 1/2. Thus $x \not\in B_t$.

    We have shown that the series $B_t$ can not outgrow $B_r$, next we show that it shrinks. Define $\mathcal H$ to be the set of tangential half-spaces of $B_r$. As before, each $H \in \mathcal{H}$ induces a half-space initialization $f_0^H$ in local coordinates, where $B_r$ sits on the right side. Notice, that each $H$ induces the same sequence of functions just over different local coordinates. To keep thinks simple, we will think of each $f^H$ to be always evaluated over their local coordinate system, and $g$ over the global coordinate system. Whenever we relate the two, you may assume that $z = A_H x$, where $A_H$ is the system induced by $H$. Naturally 
    \begin{align*}
        g_0 \le f_0^H
    \end{align*}
    and by Lemma~\ref{lem:T-monotone}.\ref{lem:T_mon_2}
    \begin{align*}
        g_1 \le f_1^H
    \end{align*}
    for all $H \in \mathcal{H}$. Next we show that each additional mean-field update introduces an additional error into $g$. Specifically, we show that
    \begin{align*}
        g_t(w, x) \le f_t^H\big(w, z-(t-1)E(W)\big)
    \end{align*}
     for all $H \in \mathcal{H}$ and an error $E$ uniform in $w$ and $z$. We show this proposition by induction on $t\ge 2$. Consider any point $x$ in global coordinates and let $y$ be the point on the boundary of $B_r$ farthest away from $x$. Consider the ball of radius $k^{1/d}$ centred at $y$, at least half $B$ of this ball lies on the left of the half-space $H_y \in \mathcal H$ going through $y$. The volume of that half is $\Theta(k)$, furthermore all vertices in it are at most $ \Theta(n^{1/d})$ away from $x$ (this is a global constraint in a $d$-dimensional cube of volume $n$) and thus all vertices of weight at least
    \begin{align*}
        w= \Theta(n)/k
    \end{align*}
    connect to all other vertices in the space. We now compare how much the space of $B$ contributes to the advantage of $x$ in $f_t^{H_y}$ and $f_t^{\overline{H_y}}$, where $\overline{H_y} \in \mathcal{X}$ is the half-space going through the antipodal point of $y$. In $H_y$, $B$ is treated as just to the left of the half-space, and thus by Lemma~\ref{lem:valid_f}, $f_1^{H_y}<1/2$ in the full domain. In $\overline{H_y}$, the region is treated as if $z\ge r$ and thus bounded by Theorem~\ref{thm:main} strictly away from $1/2$ in the opposite direction. We define the error as the contribution difference to the advantage of $B$ for the two perspectives
    \begin{align*}
        \delta =  \mu_{f_t^{\overline{H_y}}}^B(\cdot, x) -\mu_{f_t^{H_y}}^B(\cdot, x) = \Theta(1) \, \mathrm{Vol}(B) \, n^{1-\tau}\, \varepsilon.
    \end{align*}
    By induction hypothesis $g_t$ is upper bounded by shifted version of all functions $f_t^H$ for $H \in \mathcal{H}$, so in particular upper bounded by the minimum of $f_t^{H_y}$ and $f_t^{\overline{H_y}}$
     \begin{align*}
         g_{t+1}(w, x) &=\Phi\!\left(\frac{\mu_{g_t}(w,x)}{\sqrt{2\,\lambda(w)}}\right) \\
         &\le \Phi\!\left(\frac{\mu_{\min\left\{f_t^{H_y}, f_t^{\overline{H_y}}\right\}}\big(w,z-(t-1)E(W)\big)}{\sqrt{2\,\lambda(w)}}\right).
     \end{align*}
     Remembering the region $B$, this can be upper bounded by 
    \begin{align*}
         g_{t+1}(w, x) &\le \Phi\!\left(\frac{\mu_{f_t^{\overline{H_y}}}\big(w,z-(t-1)E(W)\big)-\delta}{\sqrt{2\,\lambda(w)}}\right).
    \end{align*}
    By choice of $H_y$, $\overline{H_y}$, is the half-space where the localized distance to the boundary $z$ is minimized. Naturally $\mu_{f_t^{\overline{H_y}}}$ is the minimizer over all $H$. Thus
    \begin{align*}
         g_{t+1}(w, x) &\le \Phi\!\left(\frac{\min_{H\in \mathcal H}\left\{\mu_{f_t^{H}}\big(w,z-(t-1)E(W)\big)-\delta\right\}}{\sqrt{2\,\lambda(w)}}\right) \\
         &\le \min_{H\in \mathcal H}\left\{\Phi\!\left(\frac{\mu_{f_t}\big(w,z-(t-1)E(W)\big)-\delta}{\sqrt{2\,\lambda(w)}}\right)\right\}
    \end{align*}
    Plugging in the definition of the advantage 
    \begin{align*}
        \mu_{f_t}(w, z) -\delta&= \int_{\mathcal{D}, \mathcal X}\big(2f_t(w',z')-1\big)\,\mathbbm{1}\big((w,z),(w',z')\big)-\delta/n\,d\eta(w',z') \\
        &\le  \int_{\mathcal{D}, \mathcal X}\big(2(f_t(w',z')-\delta/n)-1\big)\,\mathbbm{1}\big((w,z),(w',z')\big)\,d\eta(w',z').
    \end{align*}
    Differentiating $f_{t}$,
    \begin{align*}
        \bigl|\partial_z f_{t}(w,z)\bigr| =
\varphi\!\Big(\tfrac{\mu_{f_{t-1}}(w,z)}{\sqrt{2\lambda(w)}}\Big)\,
\frac{\bigl|\partial_z\mu_{f_{t-1}}(w,z)\bigr|}{\sqrt{2\lambda(w)}}
\le
\frac{\bigl|\partial_z\mu_{f_{t-1}}(w,z)\bigr|}{\sqrt{4\pi\,\lambda(w)}} = \frac{\bigl|\partial_z\lambda_{1,f_{t-1}}(w,z)\bigr|}{\sqrt{\pi\,\lambda(w)}},
    \end{align*}
as $\phi$ is upper bounded by $1/\sqrt{2\pi}$. Shifting $z$
by some offset $h$ only translates the ball of influence of radius $R(w,w')=(kww')^{1/d}$, hence the change in $\lambda_{1,f_t}$
is bounded by the volume of the symmetric difference of two radius-$R$ balls shifted by $h$,
which is at most $O (R^{d-1}|h|)$. Dividing by $|h|$ and integrating over $w'$ gives
\[
|\partial_z\lambda_{1,f_{t-1}}(w,z)|
\;=\;
O\!\left(\int (kww')^{\frac{d-1}{d}}\,\rho(w')\,dw'\right)
=
O\!\left(w^{\frac{d-1}{d}}\right),
\]
since $\int (w')^{(d-1)/d}\rho(w')\,dw'<\infty$ for $\tau>2$. Therefore
$|\partial_z\mu_{f_t}(w,z)| = O\!\big(w^{\frac{d-1}{d}}\big)$, and using $\lambda(w)=\Theta(w)$ yields
\[
|\partial_z f_{t}(w,z)|
=
O\!\left(\frac{w^{\frac{d-1}{d}}}{\sqrt{w}}\right)
=
O\!\left(w^{\frac{d-1}{d}-\frac12}\right),
\]
with constants independent of $t$. Then for all $w \le W$, $f_t(w,z)$ is $L_W$ Lipschitz for $L_W \in O(W^{\frac{d-1}{d}-\frac{1}{2}})$. Thus there exists a function $E(W)$ uniform in $w, z$ and $t$ such that 
\begin{align*}
    f_t(w, z)-\delta/n \le f_t(w, z-E(W)).
\end{align*}
Plugging this back into the bound for $g_t(w, x)$
\begin{align*}
         g_{t+1}(w, x) &\le \min_{H\in \mathcal H}\left\{\Phi\!\left(\frac{\mu_{f_t}\big(w,z-t\,E(W)\big)}{\sqrt{2\,\lambda(w)}}\right)\right\} \\
         &= \min_{H\in \mathcal H}\left\{f_{t+1}\big(w,z-t\,E(W)\big)\right\}
    \end{align*}
which completes the induction. 

We can now prove the result. Let $x \in B_r$ and let $z$ be the distance from $x$ to the boundary of $B_r$. Then for $t_0=z/E(W)$,
\begin{align*}
         g_{t_0+2}(w, x) &\le  f_{t_0+2}\big(w,z-(t_0+1)\,E(W)\big) \\
         &= f_{t_0+2}\big(w,-E(W)\big) \\
         &< 0,
    \end{align*}
and thus $x \not \in B_{t_0+2}$. By the same argument $x$ never again enters the set $B_{t}$ for any $t>t_0+2$. Since this holds for all $x \in B_r$, $B_\infty$ must be empty.
\end{proof}
\noindent The result from \ref{thm:non_convergence} shows that the missing symmetry of the ball initialization erodes any chances of a mean-field fix point $g^*$, in which both opinions survive. This result does not however rule out survival for an unbounded number of steps. In fact, as we will see in the next sections, such a result is achievable.

\subsubsection{Bounding Local Curvature}

\noindent From this point forward, we will only use the local coordinate system of $v$. It is thus convenient to define the new position of all previously defined objects in this new coordinate system. Let $X \subseteq \mathcal{X}$ be any previously defined object, by abuse of notation define $X := A_v X$, where $A_v$ is the transformation from Lemma~\ref{lem:local_coordinates}. This in particular redefines $x_v := Ax_v$ and $p :=Ap = 0$. 

In the new coordinate system, $p$ is still the closest point to $x_v$ on the boundary of $B_r$ as $A_v$ preserves distances. As $x_v-p$ is orthogonal to $H_p$, it holds that $\lVert x_v-H_p \rVert = \lVert x_v-p \rVert = \lVert x_v\rVert = |z_v|$. In particular, if $x_v \in B_r$ then $\lVert x_v-H_p \rVert =z_v$. We can now formalize the intuition from before: from the perspective of $v$, the local boundary of $B_r$ looks identical to $H_p$.
\begin{Lemma} \label{lem:bounding_curvature}
    Let $B_r \subseteq \mathcal{X}$ be a ball of radius $r \in \omega(1)$, let $v\in \mathcal{X}$ and let $H_p$ be the corresponding half-space as defined above. Let $x \in \mathcal{X}$ be some point with 
    \begin{align*}
        r_\mathrm{max} := k^{1/d}\, r^{(1-\varepsilon)/2}\ge \lVert x_v-x\rVert,
    \end{align*}
    for some $\varepsilon>0$. Let $p_x$ be the closest point on the boundary of $B_r$ to $x$ and let $h_x$ be the closest point on $H_p$ to $x$. Then 
    \begin{align*}
        \left|\lVert x-p_x\rVert-\lVert x-h_x\rVert \right|\le \Delta
    \end{align*}
    for $\Delta \in o(1)$. 
\end{Lemma}
\begin{proof}  
Let $p_{h_x}$ be the closest point on $B_r$ from $h_x$. We prove the statement by case distinction. Assume $\lVert x-p_x\rVert-\lVert x-h_x\rVert \ge 0$, then by choice of $p_x$ and the triangle inequality over the point $h_x$,
\begin{align*}
    \left|\lVert x-p_x\rVert-\lVert x-h_x\rVert \right| &= \lVert x-p_x\rVert-\lVert x-h_x\rVert \\
    & \le \lVert x-p_{h_x}\rVert-\lVert x-h_x\rVert \\
    & \le \lVert h_x-p_{h_x}\rVert.
\end{align*}
For the second case assume $\lVert x-h_x\rVert - \lVert x-p_x\rVert \ge 0$. Then by triangle inequality over the point $p_{h_x}$ and choice of $p_x$
\begin{align*}
    \left|\lVert x-h_x\rVert-\lVert x-p_x\rVert \right| &= \lVert x-h_x\rVert-\lVert x-p_x\rVert \\
    & \le \lVert h_x-p_{h_x}\rVert + \lVert x-p_{h_x}\rVert - \lVert x-p_x\rVert \\
    & \le \lVert h_x-p_{h_x}\rVert. 
\end{align*}
Define $\Delta = \lVert h_x-p_{h_x}\rVert$ and note that $\left|\lVert x-p_x\rVert-\lVert x-h_x\rVert \right|\le \Delta$. By assumption, $\lVert x_v-x\rVert \le r_\mathrm{max}$, so in particular also $\lVert x_v[1,\dots, d-1]-x[1,\dots, d-1]\rVert \le r_\mathrm{max}$ and $\lVert h_x\rVert \le r_\mathrm{max}$. This lets us construct the right triangle consisting of the points $p$, $h_x$ and the centre of $B_r$. By definition, the distance from the centre of $B_r$ to $p$ is $r$, the distance from $p$ to $h_x$ is at most $r_\mathrm{max}$ and the line from the centre to $h_x$ goes through $p_{h_x}$. This gives the following inequality from the Pythagorean theorem
\begin{align*}
    (r+\Delta)^2 \le r^2+r_\mathrm{max}^2.
\end{align*}
Solving for $\Delta$ and applying the definition of $r_\mathrm{max}$,
\begin{align*}
    \Delta &\le r \left(\sqrt{1+\frac{r_\mathrm{max}^2}{r^2}}-1\right) \\
    &\le r \left(\sqrt{1+r^{-(1+\varepsilon)}}-1\right).
\end{align*}
Then applying the generic bound $\sqrt{1+y} \le 1+y/2$ for $y\ge 0$,
\begin{align*}
    \Delta &\le r \left(\frac{r^{-(1+\varepsilon)}}{2}\right) \\
    &\le \frac{r^{-\varepsilon}}{2} \in o(1),
\end{align*}
for $r \in \omega(1)$ and $\varepsilon >0$.
\end{proof}
\noindent The result from Lemma~\ref{lem:bounding_curvature} implies that the local geometry experienced by $v$ and all its not to far away neighbours match approximatively the geometry of a half-space, given that $r$ is large enough. The choice of $r_\mathrm{max}$ in Lemma~\ref{lem:bounding_curvature} naturally lends itself to a definition of a weight $w_\mathrm{max}$, defined in such a way that for all $w_v \le w_\mathrm{max}$ all neighbours of $v$ with weight at most $w_\mathrm{max}$ are at most $r_\mathrm{max}$ away from $v$. We define $w_\mathrm{max}$ by the largest weight that satisfies this property. Solving for $w_\mathrm{max}$ in the edge criterion yields
\begin{align}\label{eq:w_max}
    w_\mathrm{max}:= \frac{r_\mathrm{max}^{d/2}}{k^{1/2}} = r^{d (1-\varepsilon) /4}.
\end{align}

\begin{Lemma}\label{lem:speed_of_erosion}
    Let $g_0$ be the ball initialization on $B_r$ and let $f_0$ be the half-space initialization, and cut off the weights at maximum $W = w_\mathrm{max}$. Then for all $x, w$ and $t$
    \begin{align*}
        g_t(w, x) \ge f_t(w, z-t\Delta),
    \end{align*}
    where $z$ is the sign distance of $x$ to the boundary of $B_r$ and $\Delta$ is the curvature parameter from Lemma~\ref{lem:bounding_curvature}.
\end{Lemma}
\begin{proof}
    We proof this statement by induction on $t$. Clearly $g_0(w, x) = f_0(w, z)$ for all $w, x$ and $z$ being the signed distance from $x$ to the boundary of $B_r$. 
    
    Thus fix some $t\ge1$ and some vertex $v$ of type $(w_v,x_v)$. Let $z_v$ be the signed distance of $x_v$ to the boundary of $B_r$, then consider the neighbours of $v$. By assumption all weights are at most $w_{max}$ and thus all neighbours of $v$ are inside a ball $B_{r_\mathrm{max}}\!(x_v)$ around $x_v$.  For any $x \in B_{r_\mathrm{max}}(x_v)$, let $z$ be its signed distance from the boundary and let $z'$ be its coordinate in the local coordinate system of $x$, by Lemma~\ref{lem:bounding_curvature}, 
    \begin{align*}
        z \ge z'-\Delta.
    \end{align*}
   Plugging this inequality into the pointwise lower bound of the induction hypothesis, for all neighbours of $v$,
    \begin{align*}
        g_{t-1}(w, x) \ge f_{t-1}(w, z'-t\,\Delta).
    \end{align*}
    Importantly, this bounds the probability for each $(w,x)$ not by their own local coordinate system, but by that of $v$. We can thus lower bound the advantage $\mu_{g_{t-1}}(w_v, x_v)$:
    \begin{align*}
        \mu_{g_{t-1}}(w_v, x_v) &= \int_{\mathcal X}(2g_{t-1} (w,x)-1)\,\mathbbm{1}\big((w_v,x_v),(w,x)\big)\,d\eta(w,x)\\
        &\ge \int_{\mathcal X}(2f_{t-1} (w,z'-t\,\Delta)-1)\,\mathbbm{1}\big((w_v,x_v),(w,x)\big)\,d\eta(w,x).
    \end{align*}
    Shifting the the ambient space by $t \,\Delta$
    \begin{align*}
        \mu_{g_{t-1}}(w_v, x_v) &\ge \int_{\mathcal X}(2f_{t-1} (w,z')-1)\,\mathbbm{1}\big((w_v,x_v),(w,x)+t\,\Delta\big)\,d\eta(w,x)  \\
        &= \int_{\mathcal X}(2f_{t-1} (w,z')-1)\,\mathbbm{1}\big((w_v,x_v-t\,\Delta),(w,x)\big)\,d\eta(w,x)  \\
        &= \mu_{f_{t-1}}(w_v, x_v-t\,\Delta).
    \end{align*}
    Lastly, plugging this advantage bound into the definition of $\mathcal{T}$
    \begin{align*}
        g_{t}(w_v, x_v) &=\Phi\!\left(\frac{\mu_{g_{t-1}}(w_v, x_v)}{\sqrt{ 2 \, \lambda(w_v)}}\right) \\
        &\ge \Phi\!\left(\frac{\mu_{f_{t-1}}(w_v, x_v-t\, \Delta)}{\sqrt{ 2 \, \lambda(w_v)}}\right) \\
        &= f_{t}(w_v, x_v-t\, \Delta),
    \end{align*}
which concludes the induction and the proof.
\end{proof}
\noindent Our main result in this section follows immediately from this lemma, that for growing $r$ in any constant time $t$ the ball of minority opinion shrinks only by an additive $o(1)$.

\begin{Theorem}\label{thm:speed_of_erosion}
    Let $B_r$ be a ball of radius $r=\omega(1)$ and let $\Delta(r) = o(1)$ be the curvature parameter from Lemma~\ref{lem:bounding_curvature}. Cut off the weights at $w_{\mathrm{max}}$ as given in Lemma~\ref{lem:speed_of_erosion}. Then for every $t\in \mathbb{N}$, all points $x$ in the ball of radius $r-t\Delta$ concentric with $B_r$ and all $w\le w_{\mathrm{max}}$ satisfy  
    \begin{align*}
        g_t(w, x) > 1/2.
    \end{align*}
\end{Theorem}
\begin{proof}
Consider any point $x$ in the ball of radius $r-t\Delta$ concentric to $B_r$, and let $z$ be the sign distance of $x$ to the boundary of $B_r$. Then $z\ge t\Delta$. Hence, $g_t(w, x) > f_t(w, z-t\Delta)$ by Lemma~\ref{lem:speed_of_erosion}, and $f_t(w, z-t\Delta) > 1/2$ by Theorem~\ref{thm:main}.
\end{proof}

\section{Discussion}

\noindent In this work, we have investigated the majority-vote opinion dynamics on Geometric Inhomogeneous Random Graphs (GIRGs) to understand how the interplay of latent geometry and degree heterogeneity shape social influence processes \cite{bringmann2019geometric,Flache2017}. Our findings represent a significant departure from the classical understanding of coarsening dynamics observed in standard statistical physics models \cite{bray2002theory,Dornic2001}. While traditional models on regular Euclidean lattices or non-spatial random graphs typically evolve toward a global consensus where one domain is completely eliminated \cite{cooper2010multiple,lyons2017probability, cox1989coalescing,gartner2018majority,mossel2014majority}, we have demonstrated that spatial complex networks can support the long-term coexistence of competing opinions through the formation of stable boundaries. This mirrors the persistence of ideological clusters seen in real-world societies. 

Our mathematical contribution provides a rigorous explanation for this \textit{arrested coarsening} \cite{Castellano2009,bray2002theory} through a tractable mean-field model of the interface. By analysing the macroscopic limit of a planar interface, we established the existence of a stable, non-trivial limiting distribution $f^*(z)$ for the interface profile. Our results show that when the average degree is sufficiently large, the update operator $\mathcal{T}$ reinforces the existing majority on either side of the boundary. 

Furthermore, we extend the results to finite balls by bounding the effect of local curvature. Here the results are more subtle. We show that any finite ball $B_r$ of minority opinion with radius $r = \omega(1)$ will erode \emph{eventually}, but that the speed of erosion is $o(1)$. However, when we return from the mean-field model to the original graph setting, then an erosion speed of $o(1)$ is not possible. Recall that in the GIRG model, the $n$ vertices draw their position uniformly at random in a cube of volume $n$. Hence, for any fixed point on the boundary the closest vertex has expected distance $\Theta(1)$. Therefore, while in the mean-field model erosion of speed $o(1)$ can occur, in the graph setting the boundary of the ball can only withdraw in discrete steps of order $\Omega(1)$. This holds regardless of the size of the ball, even for balls of growing radius $r=\omega(1)$. Hence, a speed of $o(1)$ (or sufficiently small constant speed) in the mean-field naturally translates into speed zero for the corresponding graphs. Thus our results for finite balls in the mean-field are compatible with the experimental results showing that coarsening comes to a halt for graphs. 

Note that due to monotonicity of the process, our result extends to every region that merely \emph{contains} a sufficiently large ball. Such a region may shrink and coarsen its boundaries, but will maintain a stable core which shrinks with speed $o(1)$ that preserves the local opinion. 

For the analysis of finite balls we need to include a truncation of weights both for out positive and negative result. Note that such a truncation is natural since a ball of radius $r$ in the GIRG model does not contain vertices of arbitrary weights. However, while the truncation in Theorem~\ref{thm:non_convergence} is so large that with high probability it does not exclude any vertices of the GIRG model, the truncation in Lemma~\ref{lem:speed_of_erosion} and Theorem~\ref{thm:speed_of_erosion} is of order $r^{d(1-\varepsilon)/4}$. While it is still true that \emph{locally} (at any fixed boundary point) with high probability there are no vertices of higher weight, this is not generally true globally. The largest weight of a vertices in a ball of radius $r$ is of order roughly $r^{d/(\tau-1)}$, and this is only smaller than the cut-off if $\tau > 5$. It remains an open problem to strengthen our result to a cut-off of the weight that globally does not exclude any vertices in the initial ball.

The arrested coarsening that we show in this paper is a consequence of two aspects of the considered process. Firstly, we model conformity by the \emph{majority update rule} where nodes adopt the majority opinion among their neighbours. A classical alternative model is the Voter Model, where a node updates to the opinion of a \emph{random neighbour}. The Voter Model is closely connected to the theory of coalescing random walks and generally does not lead to stable phase boundaries~\cite{cooper2010multiple,lyons2017probability, cox1989coalescing}. As our results show, the majority update rule can allow opinions to coexist. However, this is not sufficient, as previous work has also shown that coexistence does not happen in many other graph models even under the majority dynamics~\cite{gartner2018majority,mossel2014majority}. Thus, the second crucial ingredient is the choice of the underlying graph model. In this work we choose GIRGs as an established model for social networks~\cite{bringmann2019geometric,blasius2024robust}. This model combines a heavy-tailed degree distribution with an underlying geometry which induces many properties that are also known from real social networks. Those include in particular clustering, communities, and small separators~\cite{bringmann2019geometric,lengler2017existence,kaufmann2025expanders}. It is natural that those properties foster coexistence of competing opinions, since an isolated cluster of one opinion can only be stable if every vertex has at least as many edges inside of the cluster than outside, a condition closely connected to the concepts of communities and of small edge separators. 

However, our simulations show that not all regions of minority survive: while small localized domains of a "blue" opinion are quickly eroded by the surrounding "red" majority, domains that exceed a critical size do not disappear. Instead, they settle into stable, rounded configurations where the interface between opinions becomes stationary. This transition point is highly sensitive to the power-law exponent $\tau$; as $\tau$ increases and the network becomes more localized with fewer long-range hubs, survival becomes possible at smaller scales.

\subsection{Future Research Directions}

\noindent While our current analysis focuses on the zero-temperature case to maintain analytical tractability, several avenues for future research remain:

\begin{itemize}
    \item \textbf{Non-Zero Temperature:} Our current analysis focuses on the zero-temperature case of the GIRG model. I would be very interesting to extend the analysis to include a temperature parameter $T=1/\alpha$, which would allow for long-range edges (``weak ties'') that are not directly mandated by the geometry~\cite{bringmann2019geometric}. This would help determine if thermal noise eventually overcomes the geometric stability of the interface.
    \item \textbf{Alternative Dynamics:} We have already mentioned that the classical Voter Model does not lead to stable phase boundaries, while our results shows the opposite for the Majority Vote Model. However, there are other models which interpolate between both variants. An example is the 3-Majority Dynamics, where a node updates with the majority opinion of three uniformly sampled neighbours, and the 2-Choice Dynamics, where one of the three Opinions of the 3-Majority Model is replaced by the node's own opinion. As the Majority Vote Model, coexistence quickly vanishes in some graphs~\cite{berenbrink2017ignore,ghaffari2018nearly}, but may be stable in social network models like GIRGs. The related process of bootstrap percolation has been studied for GIRGs, but only in the context of a single opinion spreading through the graph~\cite{Koch2016}. Note that for a sufficiently connected graph such processes degenerate eventually, so coexistence in this context means survival of both opinions for a super-polynomial time. 
    \item \textbf{Navigation and Spreading:} Given that GIRGs are known to be navigable and efficient for rumour spreading \cite{bringmann2017greedy, kaufmann2026rumour}, exploring the competition between a "fast" spreading rumour and a "stable" majority-vote opinion could provide insights into how misinformation persists alongside established consensus. A similar process has been analysed  in~\cite{candellero2021coexistence} for Hyperbolic Random Graphs, which is a special case of the GIRG model~\cite{bringmann2019geometric,bringmann2025average}. Alternatively, one could study two competing spreading processes, such as Competing First-Passage Percolation~\cite{antunovic2017competing}. 
\end{itemize}


\vspace{6pt} 






\section*{Author Contributions}
Conceptualization, M.B. and J.L.; methodology, M.B. and J.L.; software, M.B.; formal analysis, M.B. and J.L.; data curation, M.B.; writing---original draft preparation, M.B.; writing---review and editing, M.B. and J.L.; visualization, M.B.; supervision, J.L..; funding acquisition, J.L. All authors have read and agreed to the published version of the manuscript.

\section*{Funding}
This research was funded by the Swiss National Science Foundation, grant number 200021-232060. The first author was supported by the German Academic Scholarship Foundation.

\section*{Data Availability}
A link to the source code used to produce the experimental data can be found in this \href{https://doi.org/10.5281/zenodo.18800883}{repository}.

\section*{Acknowledgments}
We thank Konstantinos Lakis for his insightful feedback on this manuscript.

\section*{Conflicts of Interest}
The authors declare no conflicts of interest.

\section*{Abbreviations}
The following abbreviations are used in this manuscript:
\\

\noindent 
\begin{tabular}{@{}ll}
GIRG & Geometric Inhomogeneous Random Graph
\end{tabular}

\bibliography{references}

@inproceedings{Bringmann2017,
  author    = {Bringmann, Karl and Keusch, Ralph and Lengler, Johannes},
  title     = {Sampling Geometric Inhomogeneous Random Graphs in Linear Time},
  booktitle = {25th Annual European Symposium on Algorithms (ESA 2017)},
  pages     = {20:1--20:15},
  volume    = {87},
  year      = {2017}}

@article{blasius2024external,
  title={On the external validity of average-case analyses of graph algorithms},
  author={Bl{\"a}sius, Thomas and Fischbeck, Philipp},
  journal={ACM Transactions on Algorithms},
  volume={20},
  number={1},
  pages={1--42},
  year={2024},
  publisher={ACM New York, NY}
}

@article{bringmann2019geometric,
  title={Geometric inhomogeneous random graphs},
  author={Bringmann, Karl and Keusch, Ralph and Lengler, Johannes},
  journal={Theoretical Computer Science},
  volume={760},
  pages={35--54},
  year={2019},
  publisher={Elsevier}
}

@inproceedings{dayan2024expressivity,
  title={Expressivity of geometric inhomogeneous random graphs—metric and non-metric},
  author={Dayan, Benjamin and Kaufmann, Marc and Schaller, Ulysse},
  booktitle={International Conference on Complex Networks},
  pages={85--100},
  year={2024},
  organization={Springer}
}

@article{bringmann2025average,
  title={Average distance in a general class of scale-free networks},
  author={Bringmann, Karl and Keusch, Ralph and Lengler, Johannes},
  journal={Advances in Applied Probability},
  volume={57},
  number={2},
  pages={371--406},
  year={2025},
  publisher={Cambridge University Press}
}

@article{Krioukov2010,
  author  = {Krioukov, Dmitri and Papadopoulos, Fragkiskos and Kitsak, Maksim and Vahdat, Amin and Bogu{\~n}{\'a}, Mari{\'a}n},
  title   = {Hyperbolic geometry of complex networks},
  journal = {Physical Review E},
  volume  = {82},
  number  = {3},
  pages   = {036106},
  year    = {2010}
}

@article{boguna2010sustaining,
  title={Sustaining the internet with hyperbolic mapping},
  author={Bogun{\'a}, Mari{\'a}n and Papadopoulos, Fragkiskos and Krioukov, Dmitri},
  journal={Nature communications},
  volume={1},
  number={1},
  pages={62},
  year={2010},
  publisher={Nature Publishing Group UK London}
}

@article{komjathy2021penalising,
  title={Penalising transmission to hubs in scale-free spatial random graphs},
  author={Komj{\'a}thy, J{\'u}lia and Lapinskas, John and Lengler, Johannes},
  journal={Annales de l'Institut Henri Poincare (B) Probabilites et statistiques},
  volume={57},
  number={4},
  pages={1968--2016},
  year={2021},
  publisher={Institut Henri Poincar{\'e}}
}

@article{komjathy2020explosion,
  title={Explosion in weighted hyperbolic random graphs and geometric inhomogeneous random graphs},
  author={Komj{\'a}thy, J{\'u}lia and Lodewijks, Bas},
  journal={Stochastic Processes and their Applications},
  volume={130},
  number={3},
  pages={1309--1367},
  year={2020},
  publisher={Elsevier}
}

@article{cerf2024balanced,
  title={Balanced Bidirectional Breadth-First Search on Scale-Free Networks},
  author={Cerf, Sacha and Dayan, Benjamin and De Ambroggio, Umberto and Kaufmann, Marc and Lengler, Johannes and Schaller, Ulysse},
  journal={arXiv preprint arXiv:2410.22186},
  year={2024}
}

@inproceedings{kaufmann2026rumour,
  title={Rumour spreading depends on the latent geometry and degree distribution in social network models},
  author={Kaufmann, Marc and Lakis, Kostas and Lengler, Johannes and Ravi, Raghu Raman and Schaller, Ulysse and Sturm, Konstantin},
  booktitle={Proceedings of the 2026 Annual ACM-SIAM Symposium on Discrete Algorithms (SODA)},
  pages={6264--6307},
  year={2026},
  organization={SIAM}
}

@inproceedings{kaufmann2025expanders,
  title={Expanders in Models of Social Networks},
  author={Kaufmann, Marc and Lengler, Johannes and Schaller, Ulysse and Sturm, Konstantin},
  booktitle={International Workshop on Graph-Theoretic Concepts in Computer Science},
  pages={302--315},
  year={2025},
  organization={Springer}
}

@inproceedings{bringmann2017greedy,
  title={Greedy routing and the algorithmic small-world phenomenon},
  author={Bringmann, Karl and Keusch, Ralph and Lengler, Johannes and Maus, Yannic and Molla, Anisur Rahaman},
  booktitle={Proceedings of the ACM Symposium on Principles of Distributed Computing},
  pages={371--380},
  year={2017}
}

@inproceedings{kaufmann2024sublinear,
  title={Sublinear cuts are the exception in bdf-girgs},
  author={Kaufmann, Marc and Ravi, Raghu Raman and Schaller, Ulysse},
  booktitle={International Conference on Complex Networks and Their Applications},
  pages={366--377},
  year={2024},
  organization={Springer}
}

@article{lengler2017existence,
  title={Existence of small separators depends on geometry for geometric inhomogeneous random graphs},
  author={Lengler, Johannes and Todorovic, Lazar},
  journal={arXiv preprint arXiv:1711.03814},
  year={2017}
}

@article{antunovic2017competing,
  title={Competing first passage percolation on random regular graphs},
  author={Antunovi{\'c}, Ton{\'c}i and Dekel, Yael and Mossel, Elchanan and Peres, Yuval},
  journal={Random Structures \& Algorithms},
  volume={50},
  number={4},
  pages={534--583},
  year={2017},
  publisher={Wiley Online Library}
}

@article{candellero2021coexistence,
  title={Coexistence of competing first passage percolation on hyperbolic graphs},
  author={Candellero, Elisabetta and Stauffer, Alexandre},
  journal={Annales de l'Institut Henri Poincare (B) Probabilites et statistiques},
  volume={57},
  number={4},
  pages={2128--2164},
  year={2021},
  organization={Institut Henri Poincar{\'e}}
}

@inproceedings{berenbrink2017ignore,
  title={Ignore or comply? On breaking symmetry in consensus},
  author={Berenbrink, Petra and Clementi, Andrea and Els{\"a}sser, Robert and Kling, Peter and Mallmann-Trenn, Frederik and Natale, Emanuele},
  booktitle={Proceedings of the ACM Symposium on Principles of Distributed Computing},
  pages={335--344},
  year={2017}
}

@inproceedings{ghaffari2018nearly,
  title={Nearly-tight analysis for 2-choice and 3-majority consensus dynamics},
  author={Ghaffari, Mohsen and Lengler, Johannes},
  booktitle={Proceedings of the 2018 ACM Symposium on Principles of Distributed Computing},
  pages={305--313},
  year={2018}
}

@inproceedings{bierwirth2025stable,
  title={Stable Boundaries of Opinion Dynamics in Heterogeneous Spatial Complex Networks},
  author={Bierwirth, Mats and Lengler, Johannes},
  booktitle={International Conference on Complex Networks and Their Applications},
  pages={to appear},
  year={2025},
  organization={Springer}
}

@article{blasius2024robust,
  title={Robust Parameter Fitting to Realistic Network Models via Iterative Stochastic Approximation},
  author={Bl{\"a}sius, Thomas and Cohen, Sarel and Fischbeck, Philipp and Friedrich, Tobias and Krejca, Martin S},
  journal={arXiv preprint arXiv:2402.05534},
  year={2024}
}

@article{komjathy2024polynomial,
  title={Polynomial growth in degree-dependent first passage percolation on spatial random graphs},
  author={Komj{\'a}thy, J{\'u}lia and Lapinskas, John and Lengler, Johannes and Schaller, Ulysse},
  journal={Electronic Journal of Probability},
  volume={29},
  pages={1--48},
  year={2024},
  publisher={The Institute of Mathematical Statistics and the Bernoulli Society}
}

@article{kaufmann2025assortativity,
  title={Assortativity in geometric and scale-free networks},
  author={Kaufmann, Marc and Schaller, Ulysse and Bl{\"a}sius, Thomas and Lengler, Johannes},
  journal={arXiv preprint arXiv:2508.04608},
  year={2025}
}

@article{komjathy2023four,
  title={Four universal growth regimes in degree-dependent first passage percolation on spatial random graphs I},
  author={Komj{\'a}thy, J{\'u}lia and Lapinskas, John and Lengler, Johannes and Schaller, Ulysse},
  journal={arXiv preprint arXiv:2309.11840},
  year={2023}
}

@article{jorritsma2020not,
  title={Not all interventions are equal for the height of the second peak},
  author={Jorritsma, Joost and Hulshof, Tim and Komj{\'a}thy, J{\'u}lia},
  journal={Chaos, Solitons \& Fractals},
  volume={139},
  pages={109965},
  year={2020},
  publisher={Elsevier}
}

@article{Krapivsky2003,
  author  = {Krapivsky, Paul L. and Redner, Sidney},
  title   = {Dynamics of Majority Rule in Two-State Interacting Spin Systems},
  journal = {Physical Review Letters},
  volume  = {90},
  number  = {23},
  pages   = {238701},
  year    = {2003}
}

@article{Galam2002,
  author    = {Galam, Serge},
  title     = {Minority opinion spreading in random geometry},
  journal   = {The European Physical Journal B},
  volume    = {25},
  number    = {4},
  pages     = {403--406},
  year      = {2002}
}

@article{Koch2016,
  author    = {Koch, Christoph and Lengler, Johannes},
  title     = {Bootstrap percolation on geometric inhomogeneous random graphs},
  journal   = {Internet Mathematics},
  year      = {2021}
}

@article{Castellano2009,
  author  = {Castellano, Claudio and Fortunato, Santo and Loreto, Vittorio},
  title   = {Statistical physics of social dynamics},
  journal = {Reviews of Modern Physics},
  volume  = {81},
  number  = {2},
  pages   = {591--646},
  year    = {2009}
}

@article{Flache2017,
  author  = {Flache, Andreas and Mäs, Michael and Feliciani, Thomas and Chattoe-Brown, Edmund and Deffuant, Guillaume and Huet, Sylvie and Lorenz, Jan},
  title   = {Models of Social Influence: Towards the Next Frontiers},
  journal = {Journal of Artificial Societies and Social Simulation},
  volume  = {20},
  number  = {4},
  pages   = {2},
  year    = {2017}
}

@article{bray2002theory,
  title={Theory of phase-ordering kinetics},
  author={Bray, Alan J},
  journal={Advances in Physics},
  volume={51},
  number={2},
  pages={481--587},
  year={2002},
  publisher={Taylor \& Francis}
}

@article{Dornic2001,
  author  = {Dornic, I. and Chaté, H. and Chave, J. and Hinrichsen, H.},
  title   = {Critical Coarsening without Surface Tension: The Universality Class of the Voter Model},
  journal = {Physical Review Letters},
  volume  = {87},
  number  = {4},
  pages   = {045701},
  year    = {2001}
}

@article{cooper2010multiple,
  title={Multiple random walks in random regular graphs},
  author={Cooper, Colin and Frieze, Alan and Radzik, Tomasz},
  journal={SIAM Journal on Discrete Mathematics},
  volume={23},
  number={4},
  pages={1738--1761},
  year={2010},
  publisher={SIAM}
}

@article{mossel2014majority,
  title={Majority dynamics and aggregation of information in social networks},
  author={Mossel, Elchanan and Neeman, Joe and Tamuz, Omer},
  journal={Autonomous Agents and Multi-Agent Systems},
  volume={28},
  number={3},
  pages={408--429},
  year={2014},
  publisher={Springer}
}

@inproceedings{gartner2018majority,
  title={Majority model on random regular graphs},
  author={G{\"a}rtner, Bernd and Zehmakan, Ahad N},
  booktitle={Latin American symposium on theoretical informatics},
  pages={572--583},
  year={2018},
  organization={Springer}
}

@article{cox1989coalescing,
  title={Coalescing random walks and voter model consensus times on the torus in $Z^d$},
  author={Cox, J Theodore},
  journal={The Annals of Probability},
  pages={1333--1366},
  year={1989},
  publisher={JSTOR}
}

@book{lyons2017probability,
  title={Probability on trees and networks},
  author={Lyons, Russell and Peres, Yuval},
  volume={42},
  year={2017},
  publisher={Cambridge University Press}
}

@article{GeneratingGirgs, title={Efficiently generating geometric inhomogeneous and hyperbolic random graphs}, volume={10}, number={4}, journal={Network Science}, author={Bläsius, Thomas and Friedrich, Tobias and Katzmann, Maximilian and Meyer, Ulrich and Penschuck, Manuel and Weyand, Christopher}, year={2022}, pages={361–380}}
\end{document}